\documentclass{amsart}

\usepackage{amsmath,amssymb,amsthm} 
\usepackage{bbm}
\usepackage{stmaryrd}
\usepackage{graphicx}
\usepackage{txfonts}

\usepackage{hyperref}

\renewcommand{\P}{\mathbb{P}}

\renewcommand{\rho}{\varrho}

\newcommand{\tS}{\widetilde{S}}
\newcommand{\vphi}{\varphi}

\newcommand\1{0\le t\le 1}

\newcommand\N{\mathbb{N}}
\newcommand\E{\mathbb{E}}
\newcommand\R{\mathcal{R}}
\newcommand\Z{\mathbb{Z}}

\newcommand{\f}[2]{\frac{#1}{#2}}
\newcommand{\ceil}[1]{\left\lceil#1\right\rceil}

\newcommand{\abs}[1]{\left\lvert#1\right\rvert}
\newcommand{\barc}{\overline{c}}

\newcommand{\set}[1]{\left\{#1\right\}}
\newcommand{\half}{\frac{1}{2}}

\newtheorem{theorem}{Theorem}[section]  
\newtheorem{assumption}[theorem]{Assumption}

\newtheorem{corollary}[theorem]{Corollary}

\newtheorem{definition}[theorem]{Definition}
\newtheorem{lemma}[theorem]{Lemma}
\newtheorem{proposition}[theorem]{Proposition}

\theoremstyle{definition}

\newtheorem{remark}[theorem]{Remark}

\title[Utility Maximization in a Binomial Model]{Utility Maximization in a
  Binomial Model with transaction costs: a Duality Approach Based on the
  Shadow Price Process}

\begin{document}
\author[Ch. Bayer]{Christian Bayer} 
\address{Weierstrass Institute\\Mohrenstrasse 39\\10117 Berlin\\Germany}
\email{christian.bayer@wias-berlin.de}
\author[B. Veliyev]{Bezirgen Veliyev}
\address{Department of Mathematics\\University of Vienna\\Nordbergstrasse
  15\\1090 Vienna\\Austria}
\email{bezirgen.veliyev@univie.ac.at}

\thanks{We gratefully acknowledge to continued support of Walter
  Schachermayer, who introduced the problem to us and offered valuable hints
  and guidance. We are also grateful to Johannes Muhle-Karbe and Philipp
  D\"orsek for enlightening discussions.}

\maketitle

\begin{abstract}
  We consider the problem of optimizing the expected logarithmic utility of
  the value of a portfolio in a binomial model with proportional transaction
  costs with a long time horizon. By duality methods, we can find expressions
  for the boundaries of the no-trade-region and the asymptotic optimal growth
  rate, which can be made explicit for small transaction costs (in the sense
  of an asymptotic expansion). Here we find that, contrary to the classical
  results in continuous time, see Jane\v{c}ek and Shreve~[Fin.~Stoch.~8,
  2004], the size of the no-trade-region as well as the asymptotic growth rate
  depend analytically on the level $\lambda$ of transaction costs, implying a
  linear first order effect of perturbations of (small) transaction costs, in
  contrast to effects of order $\lambda^{1/3}$ and $\lambda^{2/3}$,
  respectively, as in continuous time models. Following the recent study by
  Gerhold, Muhle-Karbe and Schachermayer~[Fin.~Stoch.~2011 (online first)] we
  obtain the asymptotic expansion by an almost explicit construction of the
  shadow price process.
\end{abstract}

\section{Introduction}
\label{sec:introduction}

In this paper we consider the problem of optimal investment in a market
consisting of two assets, one risk-free asset, the bond, which, for
simplicity, is assumed to be constant in time and one stock. More precisely,
we assume that the investor wants to maximize her expected utility from final
wealth, i.e.,
\begin{equation*}
  \E\left[ U(V_T) \right] \to \max,
\end{equation*}
for a given finite horizon $T > 0$, a given utility function $U$ and,
certainly, a given initial wealth, say, $x$. Here, $V_T$ denotes the value of
the portfolio obtained by the investor at time $T$. In fact, we shall only
consider the case of the most tractable utility function, $U(x) =
\log(x)$.\footnote{It is also possible to carry out our analysis for CRRA
  utility functions of the form $U(x) = \frac{x^\gamma}{\gamma}$.} In this
framework, it is known since the seminal work of Merton in 1969~\cite{mer69}
that in a frictionless market in which the price of the risky asset follows a
geometrical Brownian motion (with drift $\mu$ and volatility $\sigma$), it is
optimal for the investor to keep the \emph{fraction} of wealth invested in the
risky asset, $\vphi_t S_t$, w.r.t.~the total portfolio wealth, $\vphi^0_t +
\vphi_t S_t$ constant equal to $\mu/\sigma^2$. In particular, this means that
the portfolio has to be constantly re-balanced. Of course, this result fully
deserves its fame, but nonetheless it mainly implies that the model of a
frictionless financial market in continuous time is not an adequate model of
reality in the context of portfolio optimization, since it gives an investment
strategy which would lead to immediate bankruptcy if applied in practice due
to the bid-ask spread. Consequently, it is essential to study the optimal
investment problem under transaction costs, a work undertaken by many authors
starting with Magill and Constantinides~\cite{mag/con76}. While actually
treating the related problem of optimizing utility from consumption, in this
work the main difference to the \emph{Merton rule} has already been
established in a heuristic way, namely that an investor optimizing his
expected utility keeps the proportion of wealth invested in the stock to total
wealth inside of a fixed \emph{interval} instead of fixed single
point. Consequently, the investor will not trade actively while the proportion
remains inside the interval, suggesting the term ``no-trade-region''. On the
other hand, when the proportion is about to leave the no-trade-region, then
the investor will trade stocks for bonds (or conversely) so as to just keep
the proportion inside the interval.

Since then, many papers in the finance and mathematical finance literature
have treated the problem of portfolio optimization under proportional
transaction costs, for instance \cite{DANO90}, \cite{SHSO94}, \cite{JASH04}
and \cite{TAKA88}, to mention some of the most influential ones on the
mathematical side. As usual for concave optimization problems, there are
essentially two approaches for the analysis: the primal approach, which, in
this case, is mostly based on the associated Hamilton-Jacobi-Bellman equation,
and the dual approach. Representatives of the former method are the works
\cite{SHSO94} and \cite{JASH04}, where the (asymptotic) first order effect of
the transaction costs to the no-trade-region was found for the
utility-from-consumption problem. An elegant formulation of the dual approach
is based on the notion of shadow prices, see Kallsen and
Muhle-Karbe~\cite{KAMK10}, and we especially mention the inspiring work of
Gerhold, Muhle-Karbe and Schachermayer~\cite{GMKS11}, where asymptotic
expansions for the no-trade-region and the asymptotic growth rate were found
in a utility-from-terminal-wealth problem. \cite{JASH04} and \cite{GMKS11}
found the characteristic result that the size of the no-trade-region is of
order $\lambda^{1/3}$, where $\lambda$ is the relative bid-ask-spread.

Almost all of the literature mentioned so far studied the effects of
market-friction in the form of proportional transaction costs in the case of
markets allowing continuous time trading, more specifically, in a
Black-Scholes model. In the context of a discrete model, the problem seems to
be less pressing, as infinite trading activities are anyway not possible,
which implies that the optimal portfolio strategy of a friction-less,
discrete-time model is, at least, admissible in a model with transaction
costs. However, also in a discrete-time market, such a portfolio will be far
from optimal. We refer to \cite{gen/jun94} for numerical experiments on the
effects of transaction costs in a utility-from-terminal-wealth problem. A
thorough analytical and numerical study of the use of dynamic programming was
done by Sass~\cite{sas05} allowing for very general structures of transaction
costs, including some numerical examples. \cite{chi/shi/she06} use the dual
approach for their analysis of the value function and the optimal strategy for
the super-replication problem of a derivative. In particular, when the
transaction costs are large enough, they show that buy-and-hold (or
sell-and-hold) strategies are optimal. In the context of super-replication,
one should also mention the recent \cite{dol/son11}. Last but not least, we
would also like to mention \cite{kus95}, where the convergence of the
super-replication cost in a binomial model with transaction costs was studied
when the binomial model converges weakly to a geometrical Brownian motion.

The goal of this paper is to derive similar asymptotic expansions of the size
of the no-trade-region and the asymptotic growth rate in the binomial
model\nocite{cox/ros/rub79}. For this purpose, we are going to use the shadow
price approach of \cite{GMKS11}, and, as common in this strand of research, we
shall restrict our attention to the problem of a long investment horizon $T
\to \infty$. We find explicit terms for the no-trade-region as well as the
asymptotic optimal growth rate when the relative bid-ask-spread $\lambda$ is
small, in the sense of asymptotic expansions in terms of $\lambda$. We find
that, contrary to the continuous case, in a binomial model the first order
effect of proportional transaction costs $\lambda$ to both the no-trade-region
and the optimal growth rate is of order $\lambda$.\footnote{In the continuous
  case, the first order effects are of order $\lambda^{1/3}$ and
  $\lambda^{2/3}$, respectively.} Economically, this marked difference can be
easily understood, as in a discrete-time model all-too-frequent trading is
already hindered by the model itself, which does not allow infinite trading
activities. Analytically, we find that the Black-Scholes model appears as a
singular limit of the family of binomial models. More precisely, let us
consider a family of binomial models with fixed horizon $T$ indexed by the
time increment $\delta$ converging weakly on path-space to a Black-Scholes
model. Then the no-trade-region depends analytically on $\lambda$ for every
$\delta > 0$, but in the limiting case the function is no longer
differentiable, implying different first order effects. Finally, we study the
convergence of the no-trade-region and the asymptotic growth rate to the
corresponding quantities in the Black-Scholes model provided that $\delta$ is
small compared to $\lambda$.

\section{Setting}
\label{sec:setting}

Let $\left(\Omega, \mathcal{F}, P\right)$ denote a probability space large
enough that we can define a binomial model $(S_t)_{t\in\N}$ with infinite time
horizon.\footnote{In fact, it would be sufficient to consider a family of
  finite probability spaces $(\Omega_T, \mathcal{F}_T, P_T)$ carrying the
  binomial model with $T$ periods for any $T \in \mathbb{N}$.} Throughout the
paper, the filtration $(\mathcal{F}_t)_{t\in\N}$ is generated by the process
$(S_t)_{t\in\N}$. For simplicity, we assume interest rates $r =
0$. Consequently, the model is free of arbitrage when $u > 1 > d$. Here, we
assume that we are given a \emph{re-combining tree}, i.e., $d = 1/u < 1$, but
allow for general $0 < p < 1$. (Recall that $S_{t+1} = uS_t$ with probability
$p$ and $S_{t+1} = d S_t$ with probability $1-p$.) While we allow for binomial
models with infinite time horizon, in general we shall consider the
restriction to a finite time horizon, i.e., $(S_t)_{t=0,\ldots, T}$. A
portfolio is given by the number $\vphi^0_t$ of bonds held at time $t$ (until
time $t+1$) and the number $\vphi_t$ of stocks.

Moreover, we also have a proportional transaction cost $\lambda,$ satisfying
$0 < \lambda <1.$ That is, for each $t\ge0$ the
bid and ask prices are given by $(1-\lambda) S_{t}$ and $S_{t}$,
respectively.

Before we go to more details about the markets with transaction costs, we
recall the log-optimal portfolio in a \emph{generalized} binomial model
without transaction costs.
\begin{proposition} 
  \label{frictionless} 
  Let $w_t$, $t=1,\ldots,T$, be a sequence of independent random variables
  taking the values $\pm 1$ with positive probabilities each and define a
  stochastic process $(S_t)_{t=0,\ldots,T}$ by some fixed value $S_0 > 0$ and
  by
  \begin{equation*}
    S_{t+1} \coloneqq
    \begin{cases}
      u_{t+1} S_t,& w_{t+1}=1,\\
      d_{t+1} S_t,& w_{t+1}=-1,
    \end{cases}
  \end{equation*}
  where $u_{t+1}>1>d_{t+1}>0$ are $\sigma(w_1, \ldots, w_t)$-measurable random
  variables and $0 \le t \le T-1$. Then the $\log$-optimizing portfolio for
  the stock-price is given in terms of the ratio $\pi_t$ of wealth invested in
  stock and total wealth at time $t$ by
  \begin{equation*}
    \pi_t \coloneqq \f{\vphi_t S_t}{\vphi^0_t + \vphi_t S_t} =
    \f{P(w_{t+1} = 1) u_{t+1} + P(w_{t+1}=-1) d_{t+1} -1}{(u_{t+1} -
      1)(1-d_{t+1})}.
  \end{equation*}
\end{proposition}
\begin{proof}
  The usual proof in the normal binomial model (see, for instance,
  \cite{shr04}) goes through without modifications. For the convenience of the
  reader, we give a short sketch. Let $p_t \coloneqq P(w_t = 1) \eqqcolon 1 -
  q_t$ and $\widetilde{p}_t \coloneqq \f{1-d_t}{u_t-d_t} \eqqcolon 1 -
  \widetilde{q}_t$. Then the state price density satisfies
  \begin{equation*}
    Z_t \coloneqq \prod_{s=1}^t \left( \f{\widetilde{p}_s}{p_s}
      \mathbf{1}_{\set{1}}(w_s) +  \f{\widetilde{q}_s}{q_s}
      \mathbf{1}_{\set{-1}}(w_s)\right),
  \end{equation*}
  since we have assumed that the interest rate is $0$. Denoting by $V_t$ the
  value of the optimizing portfolio, we obtain by Lagrangian optimization
  \begin{equation*}
    V_T = I(\lambda Z_T) = \f{1}{\lambda Z_T}, \quad \E\left[ Z_T \f{1}{\lambda
        Z_T} \right] = \E\left[Z_T I(\lambda Z_T) \right] = V_0,
  \end{equation*}
  using that $I(x) \coloneqq (U^\prime)^{-1}(x) = 1/x$. Thus,
  $\frac{1}{\lambda} = V_0$, and, by induction,
  \begin{equation*}
    V_t = \f{V_0}{Z_t}, \quad t = 0, \ldots, T.
  \end{equation*}
  On the other hand, $V_t = \vphi^0_{t-1} + \vphi_{t-1} S_t$, implying
  \begin{equation*}
    \vphi^0_t = \f{V_0}{Z_t} \f{u_{t+1}(1-d_{t+1}) - (u_{t+1} - d_{t+1})
      p_{t+1}}{(u_{t+1}-1) (1-d_{t+1})}, \quad \vphi_t S_t = \f{V_0}{Z_t}
    \f{p_{t+1} (u_{t+1}-d_{t+1}) -1 + d_{t+1}}{(u_{t+1}-1) (1-d_{t+1})},
  \end{equation*}
  which gives the formula for $\pi_t$.
\end{proof}

Next we give a formal definition of a self-financing trading strategy in the
binomial model with proportional transaction costs. Note that in a model with
transaction costs the initial position of the portfolio, i.e., before the very
first trading possibility, matters.
\begin{definition}
  A trading strategy is an adapted $\R^2$-valued process
  $(\vphi^0_t, \vphi_t)_{-1\le t \le T}$ such that $(\vphi^0_{-1},
  \vphi_{-1})=(x,0).$ It is called self-financing, if 
\begin{equation*}
 \vphi^0_t-\vphi^0_{t-1}+(\vphi_t-\vphi_{t-1}) S_t \leq 0 \mbox{ and }
\vphi^0_t-\vphi^0_{t-1}+(\vphi_t-\vphi_{t-1}) (1-\lambda) S_t \leq 0 , \  0\le
t \le T. 
\end{equation*}
  Moreover, it is called admissible, if the corresponding wealth
  process
  \begin{equation*}
    V_{t}(\varphi^0, \varphi):=\varphi^0_t +
    \varphi_{t}^{+}(1-\lambda)S_{t}-\varphi_{t}^{-}S_t, \  0\le t \le T,
  \end{equation*}
  is a.s. non-negative.
\end{definition}
In general, one would allow for portfolio process with negative values, as
long as there is a deterministic lower bound for the wealth. In a setting of
log-optimization, however, it makes sense to rule out such strategies as the
logarithm assigns utility $-\infty$ to outcomes with negative wealth.

\begin{definition}
  An admissible trading strategy $(\varphi^0_t, \varphi_t)_{-1\le t \le T}$ is
  called log-optimal on $\set{0, \ldots, T}$ for the bid-ask process
  $((1-\lambda)S, S)$ , if
  \begin{equation*}
    \E[\log(V_{T}(\psi^0, \psi))] \leq  \E[\log(V_{T}(\varphi^0, \varphi ))]
  \end{equation*}
  for all admissible trading strategies $(\psi^0, \psi).$
\end{definition}
Due to technical reasons, it is not easy to solve the above problem for finite
$T>0$, as the optimal strategy will be time-inhomogeneous. As usual in the
literature on models with transaction costs, we will instead
modify it in Definition \ref{moddef}, essentially by letting
$T\to\infty$. Here, we introduce the notion of a shadow price process, for
which we refer to \cite{KAMK10}.

\begin{definition}
  A shadow price process for $S$ is an adapted process $\tS$ such that
  $(1-\lambda)S_t \leq \tS_t \leq S_t$ for any $0\le t \le T$ and
  the log utility optimizing portfolio $({\varphi}^0, \varphi)$ for
  the frictionless market with stock price process $\tS$ exists and
  satisfies 
  \begin{gather*}
    \set{\varphi_t- \varphi_{t-1} >0} \subseteq \set{\widetilde{S}_{t} = S_t},\\
    \set{ \varphi_t- \varphi_{t-1} <0} \subseteq \set{ \widetilde{S}_{t} =
      (1-\lambda) S_t},
  \end{gather*}
  for all $0 \le t \le T$.
\end{definition}

By results from \cite{KAMK10}, \cite{KAMK11}, it is known that a shadow price
process exists and that the optimal portfolio in the frictionless market given
by the shadow price process is, in fact, also optimal in the model with
transaction costs. Indeed, the shadow price process can be seen as a solution
of the \emph{dual} optimization problem and is intimately related to the
notion of a consistent price system. For more background information on dual
methods for utility optimization in markets with transaction costs we refer to
the lecture notes~\cite{Scha11}.

\begin{definition} \label{moddef} Given a shadow price process $\tS =
  (\tS_t)_{0\le t \le T},$ an admissible trading strategy $(\vphi^0_t,
  \vphi_t)_{-1\le t \le T}$ is called log-optimal on $\set{0,\ldots, T}$
  for the \emph{modified problem} if 
  \begin{equation*}
    \E[\log(\widetilde{V}_{T}(\psi^0, \psi))] \leq
    \E[\log(\widetilde{V}_{T}(\vphi^0, \vphi ))]. 
  \end{equation*}
  for all admissible trading strategies $(\psi^0, \psi),$ where
  \begin{equation*}
    \widetilde{V}_{t}(\vphi^0, \vphi) \coloneqq \vphi^0_t + \vphi_{t} \tS_t, \
    t \geq 0.
  \end{equation*}
\end{definition}

\begin{proposition} \label{SOLMODPROB}
  Let $\tS$ be a shadow price process for the bid-ask price process
  $((1-\lambda)S, S)$ and let $(\vphi^0, \varphi)$ be its log-optimal
  portfolio. If $V(\vphi^0,\varphi) \geq 0$, then $(\vphi^0, \varphi)$ is
  also log-optimal for the modified problem. 
\end{proposition}
\begin{proof}
  As $\varphi$ only increases on $\set{\widetilde{S}_{t} = S_t}$ and decreases
  on $\set{ \widetilde{S}_{t} = (1-\lambda) S_t}$, we obtain that $(\vphi^0,
  \varphi)$ is self-financing for the bid-ask process $((1-\lambda)S, S).$
  Then, the assumption $V(\vphi^0,\varphi) \geq 0$ implies that $(\vphi^0,
  \varphi)$ is admissible for $((1-\lambda)S, S).$ Now, if $(\psi^0, \psi)$ is
  any admissible strategy for $((1-\lambda)S, S),$ we define a self-financing
  trading strategy $(\widetilde \psi^0, \psi)$ for the frictionless market with
  $\widetilde S$ by $\widetilde \psi_{-1}^0=x$ and $\widetilde
  \psi_{t}^0=\widetilde \psi_{t-1}^0-\widetilde S_{t}(\psi_{t}- \psi_{t-1})$ for
  $0 \leq t \leq T.$ Due to $(1-\lambda)S \leq \widetilde S \leq S$ and the
  fact that $(\psi^0, \psi)$ is admissible for $((1-\lambda)S,S),$ we obtain
  that $(\widetilde \psi^0, \psi)$ is admissible for $\widetilde S$ and
  $\widetilde \psi^0 \geq \psi^0.$ Then, we are done by
  \begin{equation*}
    \E[\log(\widetilde{V}_{T}(\psi^0, \psi))] \leq
    \E[\log(\widetilde{V}_{T}(\widetilde \psi^0, \psi))] \leq 
    \E[\log(\widetilde{V}_{T}(\vphi^0, \vphi ))].\qedhere
  \end{equation*}
\end{proof}

Using the above proposition, we obtain that difference between the true and
the modified problem is of order $\lambda$.
\begin{corollary}\label{cor:diff-true-modified}
  Let $\tS$ be a shadow price process for the bid-ask price process
  $((1-\lambda)S, S).$ \\
  (i) If its log-optimal portfolio $(\vphi^0, \varphi)$
  satisfies $\vphi^0 \geq 0$ and $\varphi \geq 0,$ then
  $$ \sup_{(\psi^0, \psi)} \E[\log({V}_{T}(\psi^0, \psi))]+ \log(1-\lambda)
  \le \E[\log({V}_{T}(\vphi^0, \vphi ))] \le  \sup_{(\psi^0, \psi)}
  \E[\log({V}_{T}(\psi^0, \psi))].$$ 
  (ii) In general, we can find a positive, bounded random variable $Y =
  Y(\lambda)$ having a finite, deterministic limit $Y(0) = \lim_{\lambda\to0}
  Y(\lambda)$ such that
  \begin{equation*}
    \sup_{(\psi^0, \psi)} \E[\log({V}_{T}(\psi^0, \psi))]+ E[\log(1-\lambda
    Y(\lambda))] 
    \le \E[\log({V}_{T}(\vphi^0, \vphi ))] \le  \sup_{(\psi^0, \psi)}
    \E[\log({V}_{T}(\psi^0, \psi))].
  \end{equation*}
\end{corollary}
\begin{proof}
  Here we only give the proof of (i). For the second part we refer to
  Lemma~\ref{lem:proof-diff-true-modified}. Let $(\psi^0, \psi)$ be any
  admissible strategy for $((1-\lambda)S,S).$ As $(1-\lambda)S \leq \widetilde
  S \leq S,$ we get ${V}_{T}(\psi^0, \psi ) \leq \widetilde{V}_{T}(\psi^0,\psi
  ).$ If $\vphi^0 \geq 0$ and $\varphi \geq 0,$ then by the same reason we
  obtain ${V}_{T}(\vphi^0, \vphi ) \geq (1-\lambda)
  \widetilde{V}_{T}(\vphi^0,\vphi ).$ Combining these with Proposition
  \ref{SOLMODPROB}, we obtain
  \begin{align*}
    \E[\log({V}_{T}(\vphi^0, \vphi ))] & \geq \E[\log(\widetilde
    {V}_{T}(\vphi^0, \vphi ))]+\log(1-\lambda) \\  
    & \geq \E[\log(\widetilde {V}_{T}(\psi^0, \psi))]+ \log(1-\lambda) \\
    & \geq \E[\log({V}_{T}(\psi^0, \psi))]+ \log(1-\lambda). \qedhere
  \end{align*}
\end{proof}
In particular, Corollary~\ref{cor:diff-true-modified} implies that both
problems coincide in the limit when $T \to \infty$. Intuitively, this is
clear, as an additional transaction at a final time $T$ should not matter much
when $T$ is large and we have a proper time-rescaling. To make this statement
precise, we need to introduce one more notion.
\begin{definition}
  \label{def:asymptotic-growth-rate}
  The \emph{optimal growth rate} is defined as
  \begin{equation*}
    R \coloneqq \limsup_{T\to\infty} \f{1}{T} \E\left[ \log\left(
        V_T(\vphi^{0,T}, \vphi^T) \right) \right], 
  \end{equation*}
  where $(\vphi^{0,T},\vphi^T)$ denotes the log-optimal portfolio for the
  time-horizon $T$.
\end{definition}
Intuitively, this means that by trading optimally, the value of the portfolio
will grow like $e^{RT}$ on average. Now,
Corollary~\ref{cor:diff-true-modified} obviously implies that we can replace
$V_T$ by $\widetilde{V}_T$ and the optimal portfolio by the optimal portfolio
of the modified problem.

\section{Heuristic construction of the shadow price process}
\label{sec:heur-constr-shad}

In this section, we are going to construct the shadow price process $\tS$ on a
heuristic level, which will then be made rigorous in the next section. In
particular, we want to stress that most of the assumptions made in this
section will be justified in Section~\ref{sec:form-constr-shad}. Moreover,
some rather heuristic and vague constructions shall be made more precise.

Following~\cite{GMKS11}, we make a particular \emph{ansatz} for the
parametrization of the shadow price process.

\begin{assumption}\label{ass:1}
  The shadow price process $\tS$ is a generalized binomial model as 
  introduced in Proposition~\ref{frictionless}.  For any \emph{excursion} of
  the shadow price process $\tS$ away from the boundaries given by the bid-
  and ask-price process, there is a deterministic function $g$ such that $\tS
  = g(S)$ during the excursion, i.e., whenever the shadow price process
  satisfies $\tS_t \in \set{(1-\lambda)S_t, S_t},\ \tS_{t+k} \in
  \set{(1-\lambda)S_{t+k}, S_{t+k}}$ but $(1-\lambda) S_{t+h} < \tS_{t+h} <
  S_{t+h}$ for any $1 \le h \le k-1$, then there is a function $g$ such that
  $\tS_{t+h} = g(S_{t+h})$, $1 \le h \le k-1$. \footnote{Note that for
    different excursions, the functions $g$ are not assumed to be equal. Later
    on, we will, however, see that those functions can be easily transformed
    into each other, see Proposition~\ref{PropBnm} together with
    Proposition~\ref{Propg}.}
\end{assumption}

We assume that we start by buying at $t=0$, i.e., $\tS_0 = S_0$. Hence, the
relation
\begin{equation*}
  x=\varphi_{0} \tS_{0} + \vphi^0_{0}
\end{equation*}
implies $\vphi^0_{0}=\frac{c x}{c+1}$ and $\varphi_{0}= \frac{x}{(c+1)
  S_{0}},$ where $c \coloneqq \frac{\vphi^0_{0}}{\varphi_{0} S_{0}}.$ Let us
note once more that $c$ is treated as a known quantity for the
moment.

In the frictionless case, Proposition~\ref{frictionless} shows that the
optimal portfolio is, indeed, determined by $c$ via $\pi = \f{1}{1+c}$. Here,
we treat the market with transaction costs as a \emph{perturbation} of the
frictionless market. Therefore, this motivates a parametrization of the
portfolio by the fraction $c$ also in that case. Keeping $c$ constant over
time requires continuous trading, incurring prohibitive transaction
costs. Consequently, we may expect that the optimal portfolio will only be
re-balanced when $c$ leaves a certain interval. Our first objective,
therefore, is to compute the initial holdings in the optimal portfolio, i.e.,
the initial $c$. In what follows, we shall, however, assume that $c$ is known
and compute the given transaction costs $\lambda$ as a function of the
parameter --- a relation, which is going to be inverted to obtain $c$.

Next, we construct the shadow price process $\tS$ during an \emph{excursion}
away from the boundary. For this, we parametrize $\tS$ not by time $t$ but by
the number $n$ of ``net upwards steps'' of the underlying price process, i.e.,
for a given $t\ge 0$, we consider $n = n(t)$ such that $S_t = u^n S_0$,
$n\in\Z$, which is possible by our choice of a re-combining binomial tree
model, i.e., by $d = u^{-1}$. During the first excursion from the bid-ask
boundary, Assumption~\ref{ass:1} implies that $\tS_t = g(S_t)$ for some
function $g$. In particular, since $n(s) = n(t)$ implies that $S_s = S_t$, we
have that $\tS_t$ will only depend on $n$, but not on time $t$
itself. Therefore, we may, during the first excursion away from the bid-ask
prices, index the shadow price process by $n$ instead of $t$.

Before constructing the shadow prices in the interior of the bid-ask price
interval, let us take a look at the expected behavior of the shadow price
process when the stock price falls, i.e., when $n \le 0$. Intuitively, and
following~\cite{GMKS11}, when the stock price gets smaller than the initial
price $S_0$, we have to continue buying stock, i.e., we have $\tS_{-n} =
S_{-n} = d^n S_0$ for $n \ge 0$, before the first instance of selling
stock.\footnote{Obviously, a positive excursion thereafter will be treated
  differently as a positive excursion immediately started at time $0$, i.e.,
  with different shadow price process.} We formulate this extended ansatz as a
second assumption.

\begin{assumption}
  \label{ass:2}
  Given a time $t\ge0$ at which the number of bonds and stocks in the
  log-optimal portfolio for the frictionless market in the shadow price
  process $\tS$ needs to be adjusted. Let $t+h$ be the (random) next time
  of an adjustment of the portfolio in the opposite direction. If $\tS_t =
  S_t$ and $S_u < S_t$, then $\tS_u = S_u$ and, conversely, if $\tS_t =
  (1-\lambda) S_t$ and $S_u > S_t$, then $\tS_u = (1-\lambda) S_u$, for
  $t \le u \le t+h$.
\end{assumption}
\begin{figure}
\begin{center}
\includegraphics[width=\textwidth]{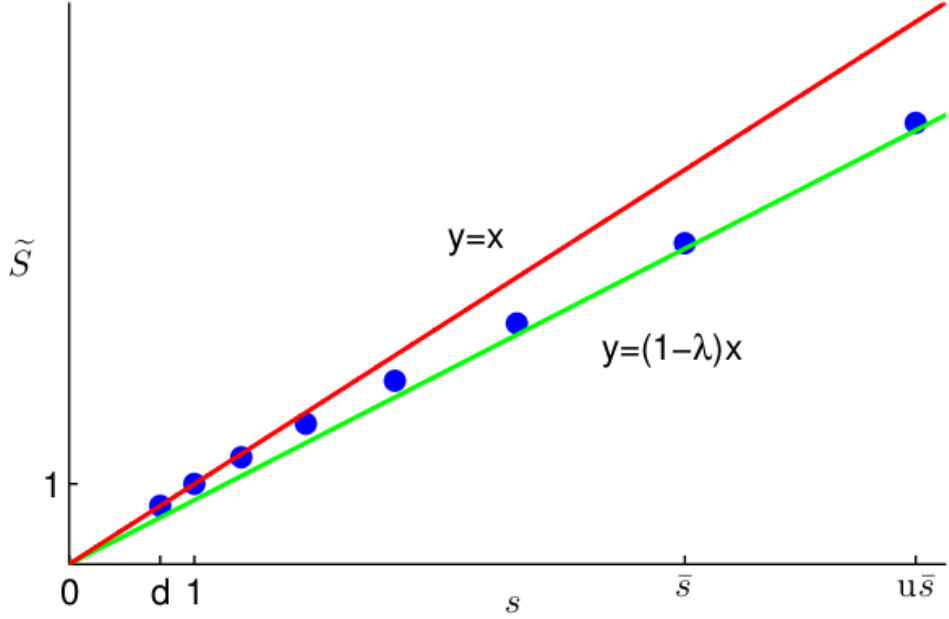}
\caption{A path of the shadow price process}
\end{center}
\end{figure}
What happens when $S_t$ increases beyond $S_0$? Intuitively, it seems clear
that we will not change the log-optimal portfolio at times $t$ with $S_t >
S_0$ except by selling stock, i.e., for positive $n$ we expect to have
$(1-\lambda) S_n \le \tS_n < S_n$. Thus, during a positive excursion of the
stock price process $S$ from $S_0$, the excursion of the shadow price process
away from the bid-ask price boundary will end at $\tau \coloneqq \min \set{t
  \ge 0\ \big|\ \tS_t = (1-\lambda) S_t}$, assuming that $\forall 0 \le t \le
\tau:\ S_t \ge S_0$. We also let $k \coloneqq n(\tau)$ be the corresponding
net number of upwards steps. This means that
\begin{equation*}
  \frac{\vphi^0_n}{\varphi_n S_0 }=c  \text{ for }  0 \leq n \leq k-1
\end{equation*}
As for $0 \le n \le k-1$ the numbers of bonds and stock in the log-optimal
portfolio for the market given by $\tS$ may not change,
Proposition~\ref{frictionless} implies
\begin{equation} \label{eq2} \pi_n= \frac{\varphi_n \tS_n}{\vphi^0_n +
    \varphi_n \tS_n} = \frac{\tS_n}{cS_{0}+ \tS_n} = \frac{p \widetilde{u}_{n+1} +
    (1-p) \widetilde{d}_{n+1} -1}{(\widetilde{u}_{n+1} - 1)(1 - \widetilde{d}_{n+1})}.
\end{equation}
where
\begin{equation*}
  \widetilde{u}_{n+1}=\frac{ \tS_{n+1}}{\tS_{n}} \text{ and } \widetilde{d}_{n+1} =
  \frac{\tS_{n-1}}{\tS_{n}}.
\end{equation*}
Solving \eqref{eq2} gives the recursion
\begin{equation*}
  \tS_{n+1} = \frac{\tS_{n}cS_{0}+p \tS_{n} \tS_{n-1} - cS_{0}(1-p) \tS_{n-1}}{pcS_{0}+
    \tS_{n-1} - (1-p)\tS_{n}},\ \tS_{0}=S_{0} \text{ and } \tS_{-1}=S_{0} d.
\end{equation*}
Fortunately, we can find an explicit solution for the above recursion. It is
given by 
\begin{align} \label{eq:tS_esplicit_p}
  \tS_n & =S_0 \frac{c(1-(\frac{1-p}{p})^n) +
    \beta_p}{-(1-(\frac{1-p}{p})^n)+\beta_p}& \mbox{ for }  p \neq
  \frac{1}{2},  \\ 
\widetilde S_{n}& =S_0 \frac{c n+ \beta}{-n+\beta} & \mbox{ for } p= \frac{1}{2} ,
\end{align}
where $\beta_p = \f{(c+d)(2p-1)}{(1-d)(1-p)}$ and $\beta=\frac{c+d}{1-d}.$

When we do not want to parametrize the shadow price process in terms of $n$,
we can still express $\tS_t = S_0 g_c(S_t)$ for $0 \le t \le \tau$. Indeed, by
$S_n = S_0 u^n$ we see that we can express $n$ in terms of the stock price $s$
by $n = \f{\log(s)}{\log(u)}$, and inserting into~\eqref{eq:tS_esplicit_p} gives
\begin{align*}
g_c(s) &= \f{c \left(1 - (\f{1-p}{p})^{-\f{\log(s)}{\log d}}
    \right) + \beta_p }{-\left(1- (\f{1-p}{p})^{-\f{\log(s)}{\log d}}\right) + \beta_p} & \mbox{ for }  p \neq \frac{1}{2}, \\
  g_{c}(s)&=\frac{c \log(s) +\beta \log u}{-\log(s) +\beta \log u} & \mbox{ for } p= \frac{1}{2}.
\end{align*}
Note that $g_{c}(s)$ is increasing, first
concave and then convex. 

Now we have constructed a candidate for the shadow price process $\tS$ which
is defined until the first time when it again hits either the bid or the ask
price of the true stock. We have also, en passant, settled the case when the
process first hits the ask price again: for $n = -1$, we have $\tS_{-1} =
S_{-1} = dS_0$, and we will buy additional stock and re-start the recursion,
but at a different initial value, see the next section for a detailed
account. However, when we actually consider the passage from ask to bid price,
i.e., when $n = k$ and $\tS_k = (1-\lambda) S_k$, we have to decide how to
re-balance our portfolio. In practice, the situation will be a bit difficult:
most likely, we are not able to follow our explicit
formula~\eqref{eq:tS_esplicit_p}, as it is quite possible that $\tS_k <
(1-\lambda) S_k$, i.e., that the recursion formula does not hold true anymore
for the last step, because it would induce a violation of the first basic
property of the shadow price process. In principle, it would be possible to
handle this situation. However, it would lead to inherent non-continuities,
which would not allow us to use the method of asymptotic expansions. Thus, we
assume that the shadow price process \emph{touches} the bid price at an
integer point $k$. (Note that this is really an assumption on the model
parameter, not just an ansatz! The assumption will be made more explicit in
Assumption~\ref{ass:k-integer} in the subsequent section.)
\begin{assumption}
  \label{ass:3}
  The model parameters ($u$, $d$, $p$, $S_0$ and $\lambda$) are chosen such
  that $\tS_k = (1-\lambda) u^k S_0$ and $\tS_{k+1} = (1-\lambda) u^{k+1}
  S_0$.
\end{assumption}
The second part of Assumption~\ref{ass:3} requires some justification. In
fact, it reflects a choice on the trading involved at the first opportunity of
selling. More precisely, it means that we \emph{do not re-balance the
  log-optimal portfolio when the shadow price process first hits the bid
price}. Only when the stock price increases once more, the shadow price is
again equal to the bid price and then we \emph{do trade}. In the discrete
time situation, this particular structure of the shadow price process seems
arbitrary, but it reflects an important condition in the continuous problem as
discussed in~\cite{GMKS11}, namely the \emph{smooth pasting condition} for the
analogous function $g$ in the Black-Scholes model with proportional
transaction costs. This condition says that $g$ is continuously differentiable
at $\overline{s}$ with $g(\overline{s}) = (1-\lambda) \overline{s}$, i.e., in
some sense the shadow price process ``smoothly'' merges with the bid price
process. In continuous time, this assumption is very beneficial in, for
instance, avoiding any reference to local times. In the discrete case, other
choices are clearly also possible, which lead, inter alia, to different shadow
price processes as the one studied by~\cite{GMKS11} in the Black-Scholes model
seen as a limiting case of the binomial model. Since one of the main
motivations for the present model is to study precisely this convergence, we
impose the second part of Assumption~\ref{ass:3}.

In the next step, we interpret the two equalities in Assumption~\ref{ass:3} as
a system of equations for the two unknowns $k$ and $\lambda$.\footnote{Recall
  that we treat $\lambda$ as an unknown and $c$ as a known quantity with the
  prospect of inverting the function for $\lambda$ in terms of $c$ at a later
  step.} For $p=\f{1}{2},$ the solution is given by
$k=\frac{(c+d)(c-1)}{c(1-d)}$ and $\lambda=1-c^2 d^k.$


For $p \neq \frac{1}{2}$, set $x = \f{1-p}{p}$ and $y = x^k$. 
If we eliminate $\lambda$ from equations, then we obtain 
\begin{equation} \label{eq:elim-lambda-for-k}
  \frac{c(1-y)+\beta_p}{-(1-y)+\beta_p}=\frac{c(1-x y)+\beta_p}{-(1-x
  y)+\beta_p} d 
\end{equation}
which is second order polynomial equation for $y.$ We obtain two solutions
$y_1$ to be given in~\eqref{eq:general-p-z} and $y_2 = \frac{p}{1-p}$ which
implies that the net number of upwards steps is $k = -1$. However, for $k=-1$,
we indeed solve equation~\eqref{eq:elim-lambda-for-k}, but at the ask-price
instead of the bid price. Therefore, the remaining solution must be the
appropriate one,
\begin{equation}
  \label{eq:general-p-z}
  y = \f{[c(p+pd-d)+d(2p-1)][1-p-pd -c(2p-1)]}{c(1-d)^2(1-p)^2}.
\end{equation}

Hence, $\bar s=d^{-k}= y^{- \log(d)/\log(x)}.$
Inserting this, we obtain
\begin{equation}
  \label{eq:general-p-lambda}
  \lambda = \frac{(c p ((c+2) d+c)-c (c+1) d) \left(-\frac{c (d-1)^2
        (p-1)^2}{((c+2) d p-(c+1) d+c p) (c (2 p-1)+d
        p+p-1)}\right)^{-\frac{\log (d)}{\log \left(x\right)}}}{c
    (2 p-1)+d p+p-1}+1 \eqqcolon F(c).
\end{equation}

\begin{remark}
  If we are mainly interested in the limit to the Black-Scholes model, we may
  assume that $u-1 \ll \lambda$. In that case, we can anyway bound
  \begin{equation*}
    \abs{\widetilde{S}_\tau - \widetilde{S}_{\ceil{\tau}}} \le (u-1)
    \widetilde{S}_\tau \ll \lambda \widetilde{S}_\tau.
  \end{equation*}
  Thus, in that sense, it should not matter for the asymptotic result, how
  we treat the boundary conditions, and whether we really hit at an integer
  point in time.
\end{remark}

\section{Formal construction of the shadow price process}
\label{sec:form-constr-shad}

The proofs of most propositions in this section are found in Appendix
\ref{ProofsAppendix}.  From now on we fix $S_0=1, 0< \lambda<1,$ $1>d=1/u>0$
and $\frac{d}{1+d}<p< \frac{1}{1+d}.$ (The last inequality translates to the
condition $0< \mu<\sigma^2$ in the Black-Scholes case. By modifying some of
the functions, it is also possible to carry out the whole analysis for the
other cases.)
Moreover, we denote $\bar c=\frac{1-p-pd}{p+pd-d}$ and $b=\frac{\log (d)}{\log
  \left((1-p)/p\right)}.$ Note that the optimal wealth fraction $\pi_t$ in the
frictionless binomial model is by Proposition~\ref{frictionless} given by
$\pi_t = \frac{1}{1+\bar{c}}$.
\begin{proposition} \label{resultsc}
Define 
$$F(c):=\begin{cases}
  1-\left (\frac{c(p+pd-d)+d(2p-1)}{(1-d)(1-p)} \right )^2
  \left(\frac{(c(p+pd-d)+d(2p-1))(1-p-p d-c (2 p-1))}
    {c (1-d)^2(1-p)^2}\right)^{b-1} & \mbox{ for } p \neq \frac{1}{2}, \\
  1-c^2 d^{\frac{(c+d)(c-1)}{c(1-d)}} & \mbox{ for } p= \frac{1}{2}.
\end{cases}$$ Then, $F(c)=\lambda$ has a unique solution in $(\bar c, \infty)$
if $p \in (\frac{d}{1+d}, \frac{1}{2}],$ and a unique solution in $(\bar c,
\frac{1-p-pd}{2p-1})$ if $p \in (\frac{1}{2}, \frac{1}{1+d}).$
\end{proposition}

As we have $c,$ we can define $k$ and $\bar s.$ Denote
\begin{equation*}
  r(c) \coloneqq \frac{[c(p+pd-d)+d(2p-1)][1-p-pd-c (2 p-1)]}{c
  (1-d)^2(1-p)^2}. 
\end{equation*}
\begin{proposition}\label{prop:define-sbar}
Fix $c$  and define
$$k:= \begin{cases}
     \frac{\log(r(c) )}{\log\left( \frac{1-p}{p}\right)}    & \mbox{ for } p
     \neq \frac{1}{2}, \\           
  \frac{(c+d)(c-1)}{c(1-d)}       & \mbox{ for } p= \frac{1}{2},
 \end{cases}$$
and $\bar s:= u^k.$ We have $k > 0.$ 
\end{proposition}

\begin{assumption}\label{ass:k-integer}
  We assume that the model parameter $d$ is given such that $k$ is a positive
  integer in the above definition.  
\end{assumption}
Note that this is the only assumption left from the previous
Section~\ref{sec:heur-constr-shad}. A closer look at the definition of $k$
shows the intuitively obvious fact that $k$ converges to infinity when $d \to
1$. Consequently, at least when we are really interested in binomial models
with $d \approx 1$, Assumption~\ref{ass:k-integer} is easy to fulfill by a
slight modification of the model parameters.

\begin{proposition} \label{Propg}
Define the function $g$ on $\{ d,1, \ldots, \bar s, u \bar s \}$ by
$$ g(s):= \begin{cases} \f{c \left(1 - (\f{1-p}{p})^{-\f{\log(s)}{\log d}}
    \right) + \beta_p }{-\left(1- (\f{1-p}{p})^{-\f{\log(s)}{\log d}}\right) + \beta_p}, & \mbox{ for }  p \neq \frac{1}{2}, \\
  \frac{c \log(s) +\beta \log u}{-\log(s) +\beta \log u}, & \mbox{ for } p= \frac{1}{2},
\end{cases}$$
where $\beta_p = \f{(c+d)(2p-1)}{(1-d)(1-p)}$ and $\beta=\frac{c+d}{1-d}.$
Then $g$ is increasing, maps $\{ d,1, \ldots, \bar s, u \bar s \}$ onto $\{d, 1, \ldots, (1-\lambda)\bar s, (1-\lambda) u \bar s \}$
and satisfies the ``smooth pasting'' conditions
\begin{equation} \label{smpaste} g(d)=d, \mbox{ } g(1)=1, \mbox{ } g(\bar
  s)=(1-\lambda)\bar s, \mbox{ } g(u \bar s)=(1-\lambda) u \bar s.
\end{equation}
In addition, $$ (1-\lambda) s \leq g(s) \leq s \mbox{ for } 1 \leq s \leq \bar s.$$
Finally, we have $$\frac{p \frac{g(us)}{g(s)} +(1-p)\frac {g(ds)}{g(s)}-1}{(\frac{g(us)}{g(s)}-1)(1-\frac{g(ds)}{g(s)})}=\frac{g(s)}{c+g(s)}
\mbox{ for } 1 \leq s \leq \bar s.$$
\end{proposition}

Define the sequence of stopping times $(\rho_n)_{n=0}^{\infty}$,
$(\sigma_n)_{n=1}^{\infty}$ and a process $(m_t)_{t \geq 0}$ by $$\rho_0=1
\mbox{ and } m_t=\min_{0 \leq i \leq t} S_i, \mbox{ } 0 \leq t \leq
\sigma_{1},$$ where $\sigma_1$ is defined as
$$ \sigma_{1}=\min \left \{ t \geq \rho_0: \frac{S_t}{m_t} =  \bar s \mbox{ }
  \& \mbox{ } \frac{S_{t-1}}{m_{t-1}} = \bar s  \right \}.$$ 
Then, define the process $(M_t)_{t \geq 0}$ as
$$M_t=\max_{\sigma_1 \leq i \leq t} S_i, \mbox{ } \sigma_1 \leq t \leq \rho_1,$$
where $\rho_1$ is defined as
$$\rho_1=\min \left \{ t \geq \sigma_1 :\frac{ S_t}{M_t} = \frac{1}{\bar s}
  \mbox{ } \& \mbox{ } \frac{ S_{t-1}}{M_{t-1}} = \frac{1}{\bar s} \right
\}.$$ 
Afterwards, we again pass to the running minimum and define
$$m_t=\min_{\rho_1 \leq i \leq t} S_i, \mbox{ } \rho_1 \leq t \leq \sigma_{2},$$
where $$ \sigma_{2}=\min \left \{ t \geq \rho_1: \frac{S_t}{m_t} = \bar s
  \mbox{ } \& \mbox{ } \frac{S_{t-1}}{m_{t-1}} = \bar s 
\right \}.$$ Then, for $t \geq \sigma_2,$ we define
$$M_t=\max_{\sigma_2 \leq i \leq t} S_i, \mbox{ } \sigma_2 \leq t \leq \rho_2,$$
where
$$\rho_2=\min \left \{ t \geq \sigma_2 : \frac{ S_t}{M_t} = \frac{1}{\bar s}
  \mbox{ } \& \mbox{ } \frac{ S_{t-1}}{M_{t-1}} = \frac{1}{\bar s} \right
\}.$$ 
 
Proceeding in a similar way, we get the stopping times
$(\sigma_n)_{n=1}^{\infty}$, $(\rho_n)_{n=1}^{\infty}$. Both $\sigma_n$ and
$\rho_n$ increase a.s.~to infinity. Note that these stopping times are indeed
attained because $S$ is a binomial model, $\bar s=u^k$ where $k \in \N$ and
$\frac{S_0}{m_0}=1, \frac{S_{\sigma_n}}{M_{\sigma_n}}=1,
\frac{S_{\rho_n}}{m_{\rho_n}}=1, \mbox{ for } n \geq 1.$ Moreover, we see that
$m_t$ and $M_t$ are only defined on stochastic intervals
$\llbracket\rho_{n-1},\sigma_n \rrbracket$ and $\llbracket\sigma_n, \rho_n
\rrbracket$ respectively.  Note that $ \bar s m_{\sigma_n-1}=S_{\sigma_n-1}
\mbox{ and } M_{\rho_n-1}=\bar s S_{\rho_n-1} \mbox{ for } n \geq 1. $ Then,
we extend the processes $M$ and $m$ to $\N$ by
$$M_t:= \bar s m_t, \mbox{ for } t \in \bigcup_{n=1}^\infty
\llbracket\rho_{n-1},\sigma_{n} \llbracket \mbox{ and } 
m_t:=\frac{M_t}{\bar s}, \mbox{ for } t \in \bigcup_{n=1}^\infty
\llbracket\sigma_n, \rho_n 
\llbracket.  $$ 
Therefore, we have $$m_t \leq S_t \leq \bar s m_t \mbox{ for } t \geq 0.$$
Furthermore, by construction, $m$ decreases only on $\{S_t=m_t  \}$ and
increases only on  
$\{ S_t= M_t \}=\{ S_t= \bar s m_t \}.$

Now, we can define a candidate for a shadow price. The result shows that it is
a generalized binomial model.
\begin{proposition}\label{PropBnm}
  Define $\widetilde S_t=m_t g(\frac{S_t}{m_t}), t \geq 0. $ Then, $\widetilde
  S$ is an adapted process which lies in the bid-ask interval
  $[(1-\lambda)S, S].$ Moreover, consider the multipliers $\widetilde{u}_t$
  and $\widetilde{d}_t$ implicitly defined by
  \begin{equation*}
    \widetilde{S}_{t+1} =
    \begin{cases}
      \widetilde{u}_{t+1} \widetilde{S}_t, & S_{t+1} = u S_t,\\
      \widetilde{d}_{t+1} \widetilde{S}_t, & S_{t+1} = d S_t,
    \end{cases}
  \end{equation*}
  then we have
  $$\widetilde u_{t+1} = \frac{g(\frac{S_t u}{m_t})}{g(\frac{S_t}{m_t})}
  >1>  
  \widetilde d_{t+1} = \frac{g(\frac{S_t d}{m_t})}{g(\frac{S_t}{m_t})}.$$
\end{proposition}
\begin{proof}
  $\widetilde S$ is adapted because $m$ is adapted.
Moreover, $$1\leq \frac{S_t}{m_t} \leq \bar s, \mbox{ for } t \geq 0.$$ Also Proposition \ref{Propg} implies that 
$$(1-\lambda) s \leq g(s) \leq s \mbox{ for } 1 \leq s \leq \bar s .$$ Hence $\widetilde S$ lies in the bid-ask interval. 
The ratios in the last assertion easily follow in the case $m_t<S_t<\bar s m_t$ as $m_t$ does not change. In the cases $S_t=m_t$ and $S_t=\bar s m_t$
they follow using $g(d)=d$ and $g(u \bar s)=(1-\lambda) u \bar s$ respectively. Finally, $\widetilde
  u_{t+1}>1>\widetilde d_{t+1},$ since $g$ is increasing.
\end{proof}

The $\log$-optimal portfolio can be given in closed form relative to the
process $m$ and the sequence of stopping times $\rho$ and $\sigma$.
\begin{theorem}\label{th1}
  Let $\widetilde S_t= m_t g \left (\frac{S_t}{m_t} \right ).$ Then the
  $\log$-optimizer $(\vphi^0_t, \varphi_t)$ in the frictionless market with
  $\widetilde S$ exists and satisfies $(\vphi^0_{-1}, \varphi_{-1})=(x,0)$,
  $(\vphi^0_{0}, \varphi_{0})=(\frac{cx}{c+1}, \frac{x}{c+1})$ and for $t >0$
 $$ \vphi^0_t= \begin{cases}  \vphi^0_{\rho_{n-1}-1} \left ( \frac{c+d}{c+1} \right )^{\frac{\log(m_t)-\log(m_{\rho_{n-1}-1})}{\log(d)}}, & \mbox{ on } 
\cup_{n=1}^{\infty} \llbracket\rho_{n-1},\sigma_n \llbracket,  \\
                          \vphi^0_{\sigma_n-1} \left ( \frac{cd+(1-\lambda)\bar s}{c+(1-\lambda) \bar s} \right )^{\frac{\log(m_t)-\log(m_{\rho_{n-1}-1})}{\log(d)}} 
\frac{m_t}{m_{\sigma_{n}}-1}, & \mbox{ on } \cup_{n=1}^{\infty} \llbracket\sigma_{n},\rho_n \llbracket,               
\end{cases} $$ 
 together with 
$$ \vphi_t= \begin{cases}  \varphi_{\rho_{n-1}-1} \left ( \frac{c+d}{c+1} \right )^{\frac{\log(m_t)-\log(m_{\rho_{n-1}-1})}{\log(d)}} 
\frac{m_{\rho_{n-1}-1}}{m_t}, & \mbox{ on } \cup_{n=1}^{\infty} \llbracket\rho_{n-1},\sigma_n \llbracket,  \\
 \varphi_{\sigma_n-1} \left ( \frac{cd+(1-\lambda)\bar s}{c+(1-\lambda) \bar s} \right )^{\frac{\log(m_t)-\log(m_{\rho_{n-1}-1})}{\log(d)}}, & \mbox{ on } 
\cup_{n=1}^{\infty} \llbracket\sigma_{n},\rho_n \llbracket.               
\end{cases} $$
Furthermore, the optimal fraction of wealth invested in the stock satisfies
$$\widetilde \pi_t= \frac{\vphi_t \widetilde S_t}{\vphi^0_t+\vphi_t \widetilde S_t}=\frac{g \left (\frac{S_t}{m_t} \right )}
{c+g \left (\frac{S_t}{m_t} \right )}. $$  
\end{theorem}
\begin{proof}
  We will show that $(\vphi^0_t, \varphi_t)$ given above is indeed the
  log-optimal portfolio.  It is clear from the above definition that
  $(\vphi^0_t, \varphi_t)$ is an adapted process. Inductively, we obtain that
\begin{equation} \label{eq5}
\vphi^0_t= c m_t \varphi_t, \mbox{ for } t \geq 0,
\end{equation}
both on $\cup_{n=1}^{\infty} \llbracket\rho_{n-1},\sigma_n \llbracket$ and on
$\cup_{n=1}^{\infty} \llbracket\sigma_{n},\rho_n \llbracket.$ Therefore, the
self-financing condition
$$\vphi^0_{t+1}-\vphi^0_t + \widetilde S_{t+1} (\varphi_{t+1}-\varphi_t)=0,$$
follows easily when $m_t$ does not change, as then $\vphi^0_t$ and $\varphi_t$
do not change, either. If $m_t$ changes and $t \in \cup_{n=1}^{\infty}
\llbracket\rho_{n-1},\sigma_n \llbracket$ , then the self-financing condition
follows using \eqref{eq5} and the fact that $\widetilde S_{t}=m_t \mbox{ and }
\widetilde S_{t+1}=m_{t+1}=d m_t.$ It follows similarly for $t \in
\cup_{n=1}^{\infty} \llbracket\sigma_{n},\rho_n \llbracket.$ Therefore,
\eqref{eq5} implies that the fraction of wealth in the stock is
$$\frac{\varphi_t \widetilde S_t}{\vphi^0_t+\varphi_t \widetilde S_t}=\frac{g \left (\frac{S_t}{m_t} \right )}
{c+g \left (\frac{S_t}{m_t} \right )}.$$ Now, we prove that the same holds for
the $\log$-optimizer and hence by uniqueness we are done.  By Proposition
\ref{PropBnm}, $\widetilde S$ is a generalized binomial model and hence
Proposition \ref{frictionless} and Proposition \ref{Propg} imply that the
fraction of wealth invested in the stock is given by
\begin{equation*}
  \widetilde \pi_t= \frac{p \widetilde u_{t+1}+(1-p) \widetilde d_{t+1} -1 }{(\widetilde
    u_{t+1}-1)(1-\widetilde d_{t+1})}= 
  \frac{p \frac{g(u\frac{S_t}{m_t})}{g(\frac{S_t}{m_t})} +(1-p)\frac
    {g(d\frac{S_t}{m_t})}{g(\frac{S_t}{m_t})}-1} 
  {(\frac{g(u\frac{S_t}{m_t})}{g(\frac{S_t}{m_t})} - 1) (1 -
    \frac{g(d\frac{S_t}{m_t})}{g(\frac{S_t}{m_t})})} = \frac{g \left
      (\frac{S_t}{m_t} \right )}{c + g \left (\frac{S_t}{m_t} \right
    )}. \qedhere 
\end{equation*}
\end{proof}

\begin{corollary}\label{cor:shadow-price-process}
Let $\widetilde S_t= m_t g \left (\frac{S_t}{m_t} \right ).$ Then $\widetilde S_t$ is
a shadow price. 
\end{corollary}
\begin{proof}
 By definition, $m$ decreases only on $\{S_t=m_t  \}$ and increases only on 
$\{ S_t= \bar s m_t \}.$ Hence, by definition of $\varphi$ in Theorem \ref{th1}, we obtain
\begin{align*}
  & \{ \varphi_t- \varphi_{t-1} >0 \}
  \subseteq \{S_t=m_t  \}=\{ \widetilde S_{t}=S_t \}  \mbox{ and } \\ &\{ \varphi_t-
   \varphi_{t-1} <0 \} 
  \subseteq \{S_t=\bar s m_t  \}=\{ \widetilde S_{t}=(1-\lambda)S_t \}. \qedhere
\end{align*}

\end{proof}

\section{Asymptotic expansions}
\label{sec:asympt-expans}

Having constructed the shadow price process and the corresponding log-optimal
portfolio process in Theorem~\ref{th1}, we can now start to reap the
benefits. Note, however, that the almost explicit account of the log-optimal
portfolio depends on the optimal ratio $c$ between wealth invested in bonds
and stocks, respectively. We have implicitly found $c$ as solution of a
non-linear equation $\lambda = F(c)$, see~\eqref{eq:general-p-lambda}, but we
need a better grip on it to facilitate further understanding of the optimal
portfolio under proportional transaction costs $\lambda$, which can be gained
by formal series expansions. In the following, denote
$\eta\coloneqq\frac{(2p-1) \log(d)}{(1-d) \log((1-p)/p)}$ if $p \neq
\frac{1}{2}$ and $\eta\coloneqq\frac{\log(d)}{-2(1-d)}$ if $p=\frac{1}{2}.$

\begin{remark}
  Assuming that we know $c$, we can find the optimal portfolio and the value
  function by a simple iteration on the tree in \emph{forward} direction,
  instead of the typical backward iteration. Thus, the shadow price method can
  be directly turned into an attractive numerical method by solving the
  equation for $c$ numerically.
\end{remark}

\begin{proposition}\label{prop:c-expansion}
  The optimal ratio of wealth invested in bonds and stocks $c$ has the series
  expansion
  $$c=\bar c + \sum_{i=1}^{\infty} c_i \lambda^i, $$ 
  where all the coefficients $c_i$ can be computed by means of well-known
  symbolic algorithms. In particular, the first two coefficients are given by
$$c_1=\frac{\bar c(1-p)}{(1+d) \eta-1} \mbox{ and } 
c_2=\frac{\bar c(1-p) \left [(1+d)^2 \eta^2+ \frac{d^2+(2-2p)d-1-2p}{1-d} \eta+\frac{2p(p+pd-d)}{1-d}\right]}{2[(1+d)\eta-1]^3}. $$
\end{proposition}
\begin{proof}
  We will try to formally invert the power series for $\lambda$ as a function of $c$.  
  Since we can only invert such a power series when the $0$-order term vanishes, we
  expand the right hand side of equation~\eqref{eq:general-p-lambda} around
  the value $c = \overline{c}=\frac{1-p-pd}{p+pd-d}$, which is the optimal $c$
  in the frictionless binomial model.

We only consider the case $p\neq \frac{1}{2},$ the case $p=\frac{1}{2}$ being similar. 
Using Mathematica~\cite{mathematica}, we do a
Taylor expansion
\begin{equation}\label{eq:lambda-p-series}
  \lambda = F(c) = \lambda_{1} (c-\overline{c}) + \lambda_2 (c-\overline{c})^2 + \mathcal{O}((c-\overline{c})^3),
\end{equation}
where $$\lambda_1 = \frac{(1+d) \eta-1}{\bar c (1-p)}, \quad 
\lambda_2= - \frac{(1+d)^2 \eta^2+ \frac{d^2+(2-2p)d-1-2p}{1-d} \eta+\frac{2p(p+pd-d)}{1-d}}{2 \bar c^2  (1-p)^2 }. $$
Note that all coefficients of the series could, in principle, be found in
symbolic form. As the first order term $\lambda_1$ does not vanish, the
implicit function theorem implies the existence of an analytic local inverse
function $F^{-1}$. The power series coefficients of the inverse function can
be found using Lagrange's inversion theorem, see, for instance,
\cite[p.~527]{knu98}. Inverting the series~\eqref{eq:lambda-p-series}, we
thus obtain obtain a series for $c$ in terms of $\lambda$
\begin{equation} \label{eq:c-p-series}
  c= \barc + c_1 \lambda + c_2 \lambda^2 + \mathcal{O}(\lambda^3), 
\end{equation} 
where $$c_1=\frac{1}{\lambda_1}=\frac{\bar c (1-p) }{ (1+d) \eta -1},
c_2=-\frac{\lambda_2}{\lambda_1^3}=\frac{\bar c(1-p) \left [(1+d)^2 \eta^2+ \frac{d^2+(2-2p)d-1-2p}{1-d} \eta+
\frac{2p(p+pd-d)}{1-d}\right]}{2[(1+d)\eta-1]^3}.$$
Again, we note that higher order coefficients can be obtained explicitly using
symbolical algorithms.
\end{proof}
\begin{remark}
  When $p\ge 1/2$, Proposition~\ref{prop:c-expansion} yields a nice economic
  interpretation. Indeed, $c_1$ is positive and increasing in $d$ and
  decreasing in $p$. Hence, the investor becomes more conservative in the
  presence of transaction costs, as $c_1 \ge 0$, and this is more pronounced
  when $d$ is large or $p$ is small, as in these cases the potential average
  gains from investment in the risky asset are relatively small. For $p<1/2$,
  the situation is less intuitive, as then the optimal fraction $c$ can become
  negative, and it does so in a singular way -- by a jump from $+\infty$
  to $-\infty$.
\end{remark}

When following the optimal strategy given in Theorem~\ref{th1}, the fraction
$\pi_t$ of the total wealth invested in the stock is kept in the interval
$[(1+c)^{-1}, (1+c/\overline{s})^{-1}]$, the \emph{no-trade region}.
\begin{theorem}
  \label{thr:no-trade-region}
  The lower and upper boundaries $\underline{\theta}$ and $\overline{\theta}$ 
  of the no-trade-region satisfy the asymptotic expansions
\begin{gather*}
    \underline \theta \coloneqq
    \frac{1}{1+c}=\frac{p+pd-d}{1-d}-\frac{(1-p-pd)(p+pd-d)(1-p)}{((1+d)
      \eta-1)(1-d)^2} \lambda + \mathcal{O}(\lambda^2),\\ 
    \overline \theta \coloneqq \frac{1}{1+c/\bar s}= \frac{p+pd-d}{1-d} +
    \frac{(1-p-pd)(p+pd-d)((1+d) \eta-(1-p))}{((1+d) \eta-1) (1-d)^2} \lambda
    + \mathcal{O}(\lambda^2)
  \end{gather*}
  for $p \neq 1/2$ and
  \begin{gather*}
    \underline \theta \coloneqq \f{1}{2} - \f{1}{4} \f{1-d}{(1+d)\log(d^{-1})
      - 2(1-d)} \lambda + \mathcal{O}(\lambda^2),\\
    \overline \theta \coloneqq \f{1}{2} + \f{1}{4}
    \f{1-d+(1+d)\log(d)}{2(1-d)+(1+d) \log(d)} \lambda + \mathcal{O}(\lambda^2)
  \end{gather*}
  for $p = 1/2$.
  The width of the no-trade-region is therefore given by
  \begin{equation*}
    \overline \theta-\underline \theta=\frac{(1-p-pd)(p+pd-d)(1+d)
      \eta}{((1+d) \eta-1) (1-d)^2} \lambda + \mathcal{O}(\lambda^2)
  \end{equation*}
  for $p \neq 1/2$ and similarly for $p = 1/2$.
\end{theorem}
\begin{proof}
  We again assume $p \neq \frac{1}{2},$ the case $p=\frac{1}{2}$ being similar. 
  We first need to compute the expansion for
  $\bar s = u^k$. Inserting the expansion for $c$ given in
  Proposition~\ref{prop:c-expansion} into the formula for $\overline{s}$ given
  in Proposition~\ref{prop:define-sbar}, we obtain
  \begin{equation*}
    \bar s=1+ s_1 \lambda +\mathcal{O}(\lambda^2), 
  \end{equation*}
  where $s_1=\frac{(1+d) \eta }{(1+d) \eta -1}$ and the further
  coefficients can, as usually, be computed using symbolic algorithms.
Then, again taking advantage of Mathematica~\cite{mathematica}, we find that the
lower boundary and the upper boundaries of the no-trade region have the
asymptotic series expansions  
  $$\underline
  \theta = \frac{1}{1+c} = \frac{p+pd-d}{1-d} -
  \frac{(1-p)(1-p-pd)(p+pd-d)}{((1+d) \eta-1)(1-d)^2} \lambda+
  \mathcal{O}(\lambda^2),$$ 
  $$\overline \theta = \frac{1}{1+c/\bar s} = \frac{p+pd-d}{1-d} +
  \frac{(1-p-pd)(p+pd-d)((1+d) \eta-(1-p))} 
  {((1+d) \eta-1)(1-d)^2} \lambda + \mathcal{O}(\lambda^2).$$
By subtracting, we get the desired formula for the width of the no-trade region.
\end{proof}
\begin{remark}
  Note that the width of the no-trade-region is positive and increasing in $d$
  to first order. This makes sense economically as larger $d$ means that the
  returns in the risky asset are smaller, so it makes sense to be more
  stringent about the transactions costs. Moreover, to first order the width
  of the no-trade-region is increasing in $p$ for $p < 1/2$ and decreasing for
  $p > 1/2$. In other words, the size of the no-trade-regions increases with
  the ``variability'' of the stock returns.
\end{remark}

Finally, we prove the second part of Corollary~\ref{cor:diff-true-modified}.
\begin{lemma}
  \label{lem:proof-diff-true-modified}
  Let $(\vphi^0, \varphi)$ be the log-optimal portfolio of the shadow-price
  process.  For $\lambda$ small enough we can find a positive, bounded random
  variable $Y = Y(\lambda)$ having a finite, deterministic limit $Y(0) =
  \lim_{\lambda\to0} Y(\lambda)$ such that
  \begin{equation*}
    \sup_{(\psi^0, \psi)} \E[\log({V}_{T}(\psi^0, \psi))]+ E[\log(1-\lambda
    Y(\lambda))] 
    \le \E[\log({V}_{T}(\vphi^0, \vphi ))] \le  \sup_{(\psi^0, \psi)}
    \E[\log({V}_{T}(\psi^0, \psi))].
  \end{equation*}  
\end{lemma}
\begin{proof}
  It is easy to see that $(1-\xi) \widetilde{V}_T(\vphi^0,\vphi) \le
  V_T(\vphi^0, \vphi)$ provided that
  \begin{equation*}
    \xi \ge \lambda \max\left( -\left(1 - \lambda + \frac{\vphi^0_T}{\vphi_T
          S_t} \right)^{-1},\ \left(1 + \frac{\vphi^0_T}{\vphi_T S_T}
      \right)^{-1} \right) \eqqcolon \lambda Y(\lambda).
  \end{equation*}
  Boundedness and positivity of $Y$ now follows from Theorem~\ref{th1} above,
  and we note that the limit for $\lambda \to 0$ is precisely given by the
  Merton proportion. The rest of the argument works just as for
  Corollary~\ref{cor:diff-true-modified}. 
\end{proof}

\section{The optimal growth rate}

In the following, we are going to consider the optimal growth rate as given in
Definition~\ref{def:asymptotic-growth-rate}.
In the frictionless binomial model, we recall from the proof of Proposition
\ref{frictionless} that the value of the log-optimal strategy satisfies
$V_T=\frac{V_0}{Z_T}$ and hence the expected 
utility is given by
$$\E[\log(V_T)]=\log(V_{0})+T \log\left(\frac{(1+d)p^p
    (1-p)^{1-p}}{d^p}\right).$$ 
Therefore, the optimal growth rate satisfies
\begin{equation}
\lim_{T \to \infty} \frac{\E[\log(V_T)]}{T} = \log\left(\frac{(1+d)p^p
    (1-p)^{1-p}}{d^p}\right).\label{eq:opt-growth-frictionless}
\end{equation}

\begin{theorem}
  \label{thr:growth-rate}
  The optimal growth rate in a binomial model with proportional transaction
  costs satisfies
  \begin{equation*}
    R = \frac{c(1-d)}{c^2-d}\log\left(\frac{(c+d)}{\sqrt{d}(c+1)}\right)
  \end{equation*}
  when $p = \frac{1}{2}$ and
  \begin{equation*}
    R =
    \frac{1-2p}{(1-p)(1-(\frac{p}{1-p})^{k+1})}
    \left[(1-p)\log\left(\frac{c+d}{c+1}\right) + p \left(\frac{p}{1-p}\right)^k
      \log\left(\frac{(c+d)p}{(c-1)(1-p)d}  \right) \right]
  \end{equation*}
  otherwise.
\end{theorem}
\begin{proof}
  We recall from Proposition \ref{PropBnm} that up and down factors for
  $\widetilde S$ are $\widetilde u_{t+1}=\frac{g(Z_t u)}{g(Z_t)}$ and
  $\widetilde d_{t+1}=\frac{g(Z_t d)}{g(Z_t)},$ where
  $Z_t\coloneqq\frac{S_t}{m_t}.$ Hence, using Proposition \ref{frictionless},
  we compute the expected $\log$-utility as
  \begin{align*}
    \E[\log(\widetilde V_T)]=&\log(\widetilde V_0)-\sum_{t=1}^T \E \left[ \log
      \left( \f{\widetilde{p}_t}{p} \mathbf{1}_{\set{1}}(w_t) +
        \f{\widetilde{q}_t}{1-p}
        \mathbf{1}_{\set{-1}}(w_t)\right) \right] \\
    =& \log(\widetilde V_0)-p\sum_{t=1}^T \E \left[ \log \left(
        \f{1-\widetilde d_t}{p\left(\widetilde u_t-\widetilde d_t \right)}
      \right) \right]
    -(1-p)\sum_{t=1}^T \E \left[ \log \left( \f{\widetilde
          u_t-1}{(1-p)\left(\widetilde u_t-\widetilde d_t \right)}\right)
    \right] \\ 
    =& \log(\widetilde V_0)-p\sum_{t=1}^T \E \left[ \log \left(
        \f{g(Z_{t-1})-g(Z_{t-1}d)}{p\left(g(Z_{t-1}u)-g(Z_{t-1}d) \right)}
      \right) \right]\\ 
    &-(1-p)\sum_{t=1}^T \E \left[ \log \left(
        \f{g(Z_{t-1}u)-g(Z_{t-1})}{(1-p)\left(g(Z_{t-1}u)-g(Z_{t-1}d)
          \right)}\right) \right].
  \end{align*}
  Now, we know from Proposition \ref{Propg} that $\frac{p g(us)
    +(1-p)g(ds)-1}{(g(us)-g(s))(g(s)-g(ds))}=\frac{1}{c+g(s)} \mbox{ for } 1
  \leq s \leq \bar s.$ Then, an elementary calculation implies
  \begin{equation*}
    \frac{g(s)-g(ds)}{p(g(us)-g(ds))}=\frac{c+g(s)}{c+g(us)}, \quad
    \frac{g(us)-g(s)}{(1-p)(g(us)-g(ds))}=\frac{c+g(s)}{c+g(ds)}.
  \end{equation*}
  Thus, using these identities we obtain that
  \begin{align*}
    R &= \lim_{T\to\infty} \left\{ -p\frac{1}{T} \sum_{t=1}^T \E \left[ \log
        \left( \f{c+g(Z_{t-1})}{c+g(Z_{t-1}u)} \right) \right]
    -(1-p)\frac{1}{T}\sum_{t=1}^T \E \left[ \log \left(
        \f{c+g(Z_{t-1})}{c+g(Z_{t-1}d)}\right) \right] \right\}  \\
    &= -p\E^{*}\left[\log\left(\frac{c+g(Z_t)}{c+g(uZ_t)}\right)\right] -(1-p)
    \E^{*}\left[\log\left(\frac{c+g(Z_t)}{c+g(dZ_t)}\right)\right],
  \end{align*}
  where the last step is due to the ergodic theorem and $\E^*$ denotes the
  expectation with respect to the invariant distribution of $Z_t$. Note that
  $Z_t$ is a Markov chain with state space $\{1,u,u^2, \ldots, u^k \}$ and
  transition matrix
  \begin{equation*}
    P_{i,j} \coloneqq \P[Z_{t+1}=u^j|Z_{t}=u^i] = 
    \begin{cases} 
      p, & j=i+1, \ 0 \le i \le k-1, \\ 
      1-p,  &j=i-1,\ 1 \le i \le k, \\
      p, & j=i=k, \\ 
      1-p, & j=i=0, \\ 
      0, & \mbox{ else. } 
    \end{cases}
  \end{equation*}
  Then the invariant distribution is the solution of $\alpha^T P=\alpha^T$
  normalized to $\sum_n \alpha_n = 1$. If
  $p=\frac{1}{2},$ the solution satisfies $\alpha_n=\frac{1}{k+1}, \mbox{ for
  } 0 \leq n \leq k.$ If $p \neq \frac{1}{2},$ we get
  $\alpha_n=\frac{1-2p}{(1-p)\left(1-\left(\frac{p}{1-p}\right)^{k+1} \right)}
  \left(\frac{p}{1-p}\right)^n, \mbox{ for } 0 \leq n \leq k.$


  For the remainder of the proof, we assume $p \neq \frac{1}{2},$ the other
  case being similar. Then, the optimal growth rate becomes 
  \begin{align*}
    R=&-p\E^{*}\left[\log\left(\frac{c+g(Z_t)}{c+g(uZ_t)}\right)\right]
    -(1-p) \E^{*}\left[\log\left(\frac{c+g(Z_t)}{c+g(dZ_t)}\right)\right] \\
    =& \E^{*}\left[\log \left
        (-\left(1-(\frac{1-p}{p})^{\frac{\log(Z_t)}{\log(u)}}\right)+\beta_p
      \right) \right]- p \E^{*}\left[ \log\left( -\left( 1-( \frac{1-p}{p}
          )^{\frac{\log(uZ_t)}{\log(u)}} \right) + \beta_p \right)\right] \\
    &-(1-p) \E^{*}\left[\log\left( -\left(1-(\frac{1-p}{p}
          )^{\frac{\log(dZ_t)}{\log(u)}}\right)+\beta_p \right)\right]\\ 
    =&\frac{1-2p}{(1-p)(1-(\frac{p}{1-p})^{k+1})}
    \Bigg[(1-p)\left(\log(\beta_p)- \log\left(\left( \frac{1-p}{p}
        \right)^{-1}+\beta_p- 1 \right)\right) \\ 
    & +p \left(\frac{p}{1-p}\right)^k \left(
      \log\left(\left(\frac{1-p}{p}\right)^{k}+\beta_p-1 \right)-
      \log\left(\left(\frac{1-p}{p}\right)^{k+1}+\beta_p-1 \right)\right)
    \Bigg] \\
    =&\frac{1-2p}{(1-p)(1-(\frac{p}{1-p})^{k+1})} \left[ (1-p)
      \log\left(\frac{c+d}{c+1} \right) + p \left(\frac{p}{1-p}\right)^k
      \log\left(\frac{(c+d)p}{(c-1)(1-p)d} \right) \right]. \qedhere
  \end{align*}
\end{proof}

Writing $k$ in terms of $c$ and plugging in the series expansion for $c$, we get
\begin{corollary}
  \label{cor:growth-rate-expansion}
  The optimal growth rate has the expansion
  \begin{align*}
    R=& \log\left(\frac{(1+d)p^p(1-p)^{1-p}}{d^p}\right) \\ & + 
    \frac{(p+pd-d)(1-p-pd)-(1+d)^2(1-p) p \log\left(\frac{(1+d)^2 (1-p)p}{d}
      \right)}{(1-d^2)\left[(1+d) \eta-1 \right]}  
    \lambda+  \mathcal{O}(\lambda^2). 
  \end{align*}
\end{corollary}
\begin{remark}
  The first order correction term in Corollary~\ref{cor:growth-rate-expansion}
  is negative, reflecting the trivial observation that transaction costs
  reduce the optimal growth rate. Moreover, contrary to the width of the
  no-trade-region, the term is decreasing in $d$ and increasing in $p$ for $p
  > 1/2$ and decreasing for $p<1/2$. Thus, the optimal growth rate is most
  effected by transactions costs, when the model is close to the Black-Scholes
  model.
\end{remark}

\section{Convergence to the Black-Scholes model}
\label{sec:conv-black-schol}

Historically, proportional transaction costs have mainly been studied in the
framework of the Black-Scholes model, see \cite{TAKA88},  \cite{DANO90}, \cite{DULU91},
 \cite{SHSO94}, \cite{CVKA96}, \cite{JASH04}, \cite{KAMK10}, \cite{GMKS11}. In
order to compare our results to previous results, we shall, therefore, obtain
a common ground for the binomial model and the Black-Scholes model. When we
add more and more periods to the binomial model while letting $u$ and $d$
approach $1$, the binomial model will clearly converge to the Black-Scholes
model. Here, we want to keep the convenient choice $d = 1/u$, while still
allowing all possible drift and volatility values $\mu$ and $\sigma$ in the
limiting Black-Scholes model. Hence, we need $p$ converging to $\half$, but
allowing it to be different from $\half$ for every finite time-step. More
precisely, if we choose a time-step $\delta \coloneqq T/N > 0$, and define 
\begin{equation}
  \label{eq:asymptotic-bs}
  p \coloneqq \frac{1}{2} + \frac{\mu-\frac{\sigma^2}{2}}{2 \sigma}
  \sqrt{\delta}, \quad d \coloneqq \exp(-\sigma \sqrt{\delta}),
\end{equation}
$u = 1/d$, then the binomial model $(S_n)_{n=0}^{N}$ with parameters $S_0 > 0$
and $p, u, d$ as above will converge to the geometrical Brownian motion
$\left(S_0 \exp\left( \sigma B_t + \left(\mu - \half \sigma^2 \right) t\right)
\right)_{t\in[0,T]}$ as $N \to \infty$ in distribution -- this is a
consequence of the invariance principle, see, e.g., Ethier and
Kurtz~\cite[Th.~7.4.1]{eth/kur86}. Moreover, the shadow price process $\tS$ of
the binomial model with proportional transaction costs $\lambda$ will also
converge to the shadow price process of the Black-Scholes model with
proportional transaction costs $\lambda$. Indeed, both shadow price processes
are parametrized by the respective functions $g_c$, in the binomial case given
in Proposition~\ref{Propg}, in the Black-Scholes case given in~\cite[Lemma
4.3]{GMKS11} as
\begin{equation*}
  g^{(BS)}_c(s) =
  \begin{cases}
    \frac{-cs+(2 \theta -1 + 2c \theta) s^{2 \theta} }{s-(2-2
      \theta-c(2\theta-1))s^{2 \theta}},& \theta \notin \set{\half,1},\\
    \frac{(c+1)+c\log(s)}{c+1-\log(s)},& \theta = \half,
  \end{cases}
\end{equation*}
where $\theta \coloneqq \frac{\mu}{\sigma^2}$, and it is an easy exercise to
verify that
\begin{equation*}
  \lim_{\delta\to0} g_c(s) = g_c^{(BS)}(s).
\end{equation*}

From equation~\eqref{eq:lambda-p-series} we can, however, see a big difference
between the binomial and the Black-Scholes case: In the binomial model, the
inverse function theorem shows that we can invert the function $F(c)$ and the
inverse function $c = G(\lambda) \coloneqq F^{-1}(\lambda)$ is analytic in a
neighborhood of $\lambda = 0$. On the other hand, we cannot directly apply the
inverse function theorem in the Black-Scholes case, since then the first and
second derivatives of $F$ at the corresponding point $\overline{c} =
\f{1-\theta}{\theta}$ vanish. Indeed, this can be seen already from the
  derivatives in the binomial model. If we plug in~\eqref{eq:asymptotic-bs}
  and do a Taylor expansion in $\delta$, then the first three derivatives of
  $F$ are
  \begin{subequations}\label{eq:lambda_asym}
    \begin{gather}
      \lambda_1 = \frac{2}{3} \theta^2 \sigma^2 \delta +
      \mathcal{O}\left(\delta^{3/2}\right),\label{eq:lambda_1_asym}\\
      \lambda_2 = \frac{2 \theta^3 \sigma}{1-\theta}
      \sqrt{\delta} +
      \mathcal{O}\left(\delta\right)\label{eq:lambda_2_asym},\\
      \lambda_3 = \frac{4 \theta^4}{3 \left(\theta -1 \right)^2}
      +O\left(\delta^{1/2}\right).
    \end{gather}
  \end{subequations}
In~\cite{GMKS11}, this problem is solved by taking the third root, i.e., by
considering the equation $\lambda^{1/3} = F(c)^{1/3}$. The power series of
$F^{1/3}$ around $c=\overline{c}$ -- corresponding to $\lambda=0$ -- then has
non-vanishing first-order term, and thus can be inverted, giving an expansion
of $c$ in terms of $\lambda^{1/3}$, see~\cite[Proposition 6.1]{GMKS11}. In
Section~\ref{sec:asympt-expans}, we have already discussed the economic
implications of this observation.

As a trivial mathematical consequence, we cannot directly obtain the series
coefficients of the relevant quantities in the Black-Scholes model as limits
of the corresponding series coefficients in the binomial model, as the former
are coefficients of a fractional power series in terms of $\lambda^{1/3}$,
whereas the latter are coefficients of an ordinary power series in
$\lambda$. Indeed, it is easy to see that the power series coefficients of,
for instance, $c$ in terms of $\lambda$ in the binomial model diverge when we
take the limit $\delta \to 0$, which is owed to the fact that the limiting
function is not analytic in $\lambda$ and, hence, does not admit a power
series expansion. 

On the other hand, we would like to stress that the quantities of interest
will actually converge to the corresponding quantities in the Black-Scholes
model when $\delta\to0$. More precisely, let us consider the optimal
wealth-fraction $c$ itself. Assuming the parameters~\eqref{eq:asymptotic-bs}
in the binomial model with fixed $\mu$ and $\sigma$, let us denote $c
\eqqcolon G(\delta, \lambda)$ when we stress the dependence on the remaining
variables $\delta > 0$ and $\lambda$. Moreover, we denote by $G(0,\lambda)$
the optimal wealth fraction $c$ in the Black-Scholes model with corresponding
parameters $\mu$ and $\sigma$. Then we obtain the
\begin{lemma}\label{lem:1}
  The function $G=G(\delta,\lambda)$ is continuous in its arguments.
\end{lemma}
For the proof we again refer to Appendix~\ref{ProofsAppendix}. To summarize,
the actual quantities of interest, like the form and size of the
no-trade-region, do converge when we approach the Black-Scholes model by a
sequence of binomial models, but their series expansions fail to converge due
to non-analyticity of the optimal wealth fraction $c$ at $\delta =
0$. Consequently, our methods cannot predict the results in the Black-Scholes
model from the corresponding results in the binomial model.



This also implies that one has to be very careful in deriving
\emph{quantitative} information from the series expansions obtained in
Sections~\ref{sec:asympt-expans} and~\ref{cor:growth-rate-expansion}. Indeed,
to get quantitative results, one needs to truncate the power
series. Unfortunately, for fixed $\lambda$, one needs to include more and more
terms of the expansion to get a similar accuracy when $\delta$ becomes
smaller.

We can, however, consider the optimal wealth proportion $\barc$ itself. In the
asymptotic regime~\eqref{eq:asymptotic-bs}, we have
\begin{equation*}
  \barc = \frac{1-\theta }{\theta } + \frac{\sigma^2\left(1-2
      \theta\right)}{24 \theta^2} \delta + \frac{\left(24 \theta^2-22
      \theta +5\right) \sigma^4}{2880 \theta^3} \delta^2 +
  \mathcal{O}\left(\delta^{3}\right),
\end{equation*}
implying that the optimal proportion $\barc$ is larger than the optimal
proportion $\barc = \f{1-\theta}{\theta}$ in the Black Scholes case if and only
if $\theta < \f{1}{2}$.



\section{A series expansion when approaching the Black-Scholes model}
\label{sec:seri-expans-when}

In Section~\ref{sec:asympt-expans} we have obtained series expansions for the
log-optimal ratios of wealth invested in the bond and wealth invested in the
stock in terms of the proportional transaction costs $\lambda$, which was
valid for ``moderate'' parameters $p, d, u$ in as much as the coefficients
diverge when $d \to 1$. Hence, these formulas are not helpful when considering
the asymptotics of the binomial model to the Black-Scholes model, see
Section~\ref{sec:conv-black-schol} above.

One possible way to obtain the series expansion of quantities of interest in
the Black-Scholes model from the related quantities in the binomial model
could be a transformation of $F$\footnote{ Recall that $c$ and $\lambda$ are
  linked by the equation $F(c) = \lambda$, with $F$ given in
  Proposition~\ref{resultsc}. Moreover, in the following we always assume that
  the parameters of the binomial model are given by~\eqref{eq:asymptotic-bs}
  with $\mu$ and $\sigma$ fixed, and, hence, we denote $F = F(\delta,c)$ with
  inverse function $G = G(\delta, \lambda)$ as above. }, in the sense that we
could try to find a mollification $\lambda(\delta)$ for $\lambda$ such that
$\lambda(0) = \lambda^{1/3}$ and $G = G(\delta, \lambda(\delta))$ is analytic
even at $\delta = 0$, but we were not successful in finding such
mollification. However, it turns out that a much simpler approach can be used
to link the asymptotic expansions for the Black-Scholes model and for the
binomial model.

Recall that the asymptotic expansion~\eqref{eq:lambda-p-series} for
$F(\delta,c)$ was of the form
\begin{equation*}
  F(\delta,c) = \sum_{i=1}^\infty \lambda_i(\delta) (c - \barc)^i
\end{equation*}
with $\lim_{\delta\to0} \lambda_1(\delta) = \lim_{\delta\to0}
\lambda_2(\delta) = 0$, but non-trivial limits for $\lambda_i(\delta)$, $i \ge
3$, cf.~\eqref{eq:lambda_asym}. Thus, for $\delta \ll 1$, we may disregard the
first two terms and instead consider
\begin{equation*}
  F_2(\delta,c) \coloneqq F(\delta, c) - \lambda_1 (c - \barc) - \lambda_2 (c
  - \barc)^2,
\end{equation*}
cf.~\eqref{eq:lambda-p-series}. By~\eqref{eq:lambda_1_asym}
and~\eqref{eq:lambda_2_asym} we see that $F(\delta, c) = F_2(\delta, c) +
g(\delta,c) \sqrt{\delta}$ for some function $g$ with non-trivial limit for
$\delta \to 0$. As before, denote the inverse function of
$F(\delta, c)$ by $G(\delta,\lambda)$ and denote the inverse function of
$F_2(\delta,c)$ by $G_2(\delta, c)$, i.e.,
\begin{equation*}
  \lambda = F(\delta, G(\delta, \lambda)) = F_2(\delta, G_2(\delta, \lambda)).
\end{equation*}
\begin{lemma}
  We have $G(\delta, \lambda) = G_2(\delta, \lambda) +
  \mathcal{O}(\delta^{1/6}+ \lambda)$ for $\delta,\ \lambda \to 0$,
  uniformly for $\delta\to0$ in $\lambda$ around $\lambda = 0$.
\end{lemma}
\begin{proof}
  By Taylor's theorem with Lagrange remainder term, we have
  \begin{align*}
    F_2(\delta, G_2(\delta, \lambda)) - F_2(\delta, G(\delta, \lambda)) &=
    \sum_{j=1}^2 \f{1}{j!} F_2^{(j)}(\delta, G(\delta,\lambda)) (G_2(\delta,
    \lambda) - G(\delta, \lambda))^j +\\
    &\quad + \f{1}{6} F_2^{(3)}(\delta, \xi)
    (G_2(\delta, \lambda) - G(\delta, \lambda))^3 \\
  \end{align*}
  for some $\xi = \xi(\delta, \lambda)$ between $G(\delta, \lambda)$ and
  $G_2(\delta, \lambda)$. At this point, let us note that $G_2(\delta,
  \lambda) - G(\delta, \lambda)$ is bounded in a neighborhood of $(0,0)$ by
  continuity. Since
  \begin{equation*}
    G(\delta, \lambda) = \overline{c} + \frac{\overline{c}(1-p)}{(1+d)\eta -1}
    \lambda + \mathcal{O}(\lambda^2),
  \end{equation*}
  we get
  \begin{align*}
    F_2^\prime(\delta, G(\delta,\lambda)) &= 3 \lambda_3(\delta)
    \left(G(\delta,\lambda) - \overline{c} \right)^2 +
    \mathcal{O}\left(\left(G(\delta,\lambda) - \overline{c} \right)^3\right)
    \\ 
    &= 3 \lambda_3(\delta) \frac{\overline{c}^2(1-p)^2}{\left( (1+d)\eta
        -1\right)^2} \lambda^2 + \mathcal{O}\left(\lambda^3\right),\\
    F_2^{\prime\prime}(\delta, G(\delta,\lambda)) &= 6 \lambda_3(\delta)
    \left(G(\delta,\lambda) - \overline{c} \right) +
    \mathcal{O}\left(\left(G(\delta,\lambda) - \overline{c} \right)^2\right)
    \\
    &= 6 \lambda_3(\delta) \frac{\overline{c}(1-p)}{(1+d)\eta -1} \lambda +
    \mathcal{O}\left( \lambda^2 \right).
  \end{align*}
  On the other hand, we also have
  \begin{align*}
    F_2(\delta, G_2(\delta, \lambda)) - F_2(\delta, G(\delta, \lambda)) &=
    F_2(\delta, G_2(\delta, \lambda)) - (F(\delta, G(\delta,\lambda)) -
    g(\delta, G(\delta, \lambda)) \sqrt{\delta})\\
    &= g(\delta, G(\delta, \lambda)) \sqrt{\delta}.
  \end{align*}
  Consequently, we get
  \begin{align*}
    \left(G_2(\delta, \lambda) - G(\delta, \lambda) \right)^3 &=
    \f{6}{F_2^{(3)}(\delta, G(\delta, \xi(\delta, \lambda)))}
    \Bigl[g(\delta,G(\delta,\lambda)) \sqrt{\delta} +\\
    & \quad + 3 \lambda_3(\delta) \frac{\overline{c}^2(1-p)^2}{\left( (1+d)\eta
        -1\right)^2} \lambda^2  \left(G_2(\delta,\lambda) - G(\delta,\lambda)
    \right) + \mathcal{O}\left(\lambda^3\right) +\\
    &\quad + 3 \lambda_3(\delta) \frac{\overline{c}(1-p)}{(1+d)\eta -1}
    \lambda \left(G_2(\delta,\lambda) - G(\delta,\lambda) \right)^2 +
    \mathcal{O}\left( \lambda^2 \right) \Bigr].
  \end{align*}
  When the $\sqrt{\delta}$-term is dominating in the right hand side, then
  this implies that $G(\delta, \lambda) = G_2(\delta, \lambda) +
  \mathcal{O}(\delta^{1/6})$. In the other two possible cases, we get
  $G(\delta, \lambda) = G_2(\delta, \lambda) + \mathcal{O}(\lambda)$, implying
  in total $G(\delta, \lambda) = G_2(\delta, \lambda) +
  \mathcal{O}(\delta^{1/6} + \lambda)$.
  
  Regarding the uniformity in $\lambda$, note that $F_2^{(3)}(\delta,
  G(\delta, \xi(\delta, \lambda))$ is bounded for $0\le \delta \le \delta_0$
  and $0 \le \lambda \le \lambda_0$ and $g(\delta, G(\delta, \lambda))$
  converges to some finite, non-zero value.
\end{proof}

By construction, we will obtain an asymptotic expansion of the form
\begin{equation*}
  G_2(\delta, \lambda) \approx \sum_{i=1}^\infty b_i(\delta) \lambda^{i/3}
  \approx G(\delta, \lambda) + \mathcal{O}(\delta^{1/6} + \lambda).
\end{equation*}
So, the coefficients $b_1(\delta)$ and $b_2(\delta)$ of $G_2$ will be
asymptotically (for $\delta \to 0$) equal to the corresponding coefficients of
$G(0, \lambda)$. In particular, if we only want to match the first coefficient
$b_1$, we have to choose $\delta \ll \lambda^2$.

This approach allows us to compare the results of the binomial model with the
results of the Black-Scholes model, at least provided that $\delta$ is small
enough when compared with $\lambda$. Let us exemplify the procedure for the
boundaries and the size of the no-trade-regions, which have been calculated in
Theorem~\ref{thr:no-trade-region} for the binomial model and
in~\cite[Corollary 6.2]{GMKS11} for the Black-Scholes model.
For the asymptotics of the optimal growth, we need to consider
$\frac{R}{\delta}$ instead of $R$, as the calender time is given in terms of
the number of periods $T$ in the binomial model by $T\delta$. We have

\begin{theorem}
  \label{thr:no-trade-region-asymptotic}
  Consider a family of binomial model with parameters $p$ and $d$ given
  by~\eqref{eq:asymptotic-bs} for fixed $\mu$, $\sigma>0$, $u=1/d$ and
  proportional transaction costs $\lambda$, which we assume to be small but
  much larger than $\delta$, at least $\lambda^2 \gg \delta$. The lower and
  upper boundaries of the no-trade-region satisfy
  $\underline{\theta}=\underline{\theta}_{0} + \underline{\theta}_1
  \lambda^{\frac{1}{3}} + \mathcal{O}\left( \lambda^{2/3}\right)$ and
  $\overline{\theta} = \overline{\theta}_{0}+\overline{\theta}_1
  \lambda^{\frac{1}{3}} + \mathcal{O}\left( \lambda^{2/3}\right)$,
  respectively, with
  \begin{gather*}
    \underline{\theta}_{0} = \theta +  \frac{\sigma^2}{24} (2\theta-1)\delta +
    \mathcal{O}\left(\delta ^{3/2}\right), \\
    \underline{\theta}_{1} = -\left ( \frac{3 \theta ^2 (1-\theta)^2}{4}
    \right)^\frac{1}{3} + \left ( \frac{3 \theta ^2 (1-\theta)^2}{32}
    \right)^\frac{1}{3} (4 \theta-3) \sigma \delta^\f{1}{2} +
    \mathcal{O}\left(\delta\right), \\
    \overline{\theta}_{0} =
    \theta  +  \frac{\sigma^2}{24} (2\theta-1)\delta + \mathcal{O}\left(\delta
      ^{3/2}\right), \\
    \overline{\theta}_{1} = \left ( \frac{3 \theta ^2 (1-\theta)^2}{4}
    \right)^\frac{1}{3} + \left ( \frac{3 \theta ^2 (1-\theta)^2}{32}
    \right)^\frac{1}{3} \sigma \delta^\f{1}{2} +
    \mathcal{O}\left(\delta\right).
  \end{gather*}
  Moreover, the width $\overline{\theta} - \underline{\theta}$ of the
  no-trade-region is given by
  \begin{equation*}
    \left( 6 \theta^2 (1-\theta)^2 \right)^\frac{1}{3} \left ( 1 +
      (1-\theta) \sigma \delta^\f{1}{2} + \mathcal{O}\left(\delta\right)
    \right) \lambda^\frac{1}{3}+\mathcal{O}\left( \lambda^{2/3}\right).
  \end{equation*}
  The asymptotic optimal growth rate satisfies $\frac{R}{\delta}=R_{0}+ R_{1}
  \lambda^{\frac{1}{3}} + \mathcal{O}\left( \lambda^{2/3}\right),$ 
  where 
  \begin{gather*}
    R_{0}=\frac{\mu^2}{2 \sigma^2}+\frac{\sigma^2}{24} \theta(\theta-1)(2
    \theta^2-2\theta+1)\delta + 
    \mathcal{O}\left(\delta ^{3/2}\right),\\
    R_1=\left(\frac{3}{32}\right)^\frac{1}{3}\sigma^3
    (\theta(\theta-1))^\frac{5}{3}\delta ^{1/2} + 
    \mathcal{O}\left(\delta\right).
  \end{gather*}
\end{theorem}
Note that in all of the above terms, the zero-order term in $\delta$ is equal
to the corresponding term in the Black-Scholes model, which again justifies
our approach. Interestingly, lowest order effect of discrete time seems to be
a shift of the no-trade region. The zero-order terms of both
$\underline{\theta}$ and $\overline{\theta}$ in the binomial model only differ
from the corresponding terms in the Black-Scholes model by the term
$\frac{\sigma^2}{24} (2\theta - 1)$, i.e., the no-trade-region is simply
shifted by that term, which is positive when $\theta > \half$ and negative
when $\theta < \half$. Consequently, the size of the no-trade-region is not
effected by discrete time at lowest order. At order $\lambda^{1/3}$, however,
the width of the no-trade region in the binomial model is larger than the size
of the no-trade-region in the binomial model by the term $6\theta^2
(1-\theta)^2 (1-\theta) \sigma \sqrt{\delta} \lambda^{1/3}$ plus higher order
term. This observation seems to be counter-intuitive, as the time-discreteness
should actually lead to a smaller no-trade-region, as an infinite variation of
trading can anyway not be accumulated since there are only finitely many
possible trading times -- which is also reflected by the results when we do
not consider $\delta\to0$. This indicates that the truncation $F \mapsto F_2$
over-compensates for the effects of continuous-time-trading.
\begin{proof}[Proof of Theorem~\ref{thr:no-trade-region-asymptotic}]
  Define $\tilde \lambda= \lambda- \lambda_1(c- \bar c)-\lambda_2 (c-\bar
  c)^2.$ Then,
  $$\tilde\lambda=\lambda_3 (c-\bar c)^3 + \lambda_4 (c-\bar c)^4+\lambda_5
  (c-\bar c)^5+\mathcal{O}((c-\bar c)^6)).$$
  Inverting the series using Lagrange's theorem, we obtain
  $$c=\bar c+ (\frac{1}{\lambda_3})^{\frac{1}{3}} \tilde
  \lambda^{\frac{1}{3}}-\frac{1}{3}
  \left(\left(\frac{1}{\lambda_3}\right){}^{5/3}  
    \lambda_4\right) \tilde \lambda^{2/3}+\frac{\left(\lambda_4^2-\lambda_3
      \lambda_5\right) \tilde \lambda}{3 \lambda_3^3}+\mathcal{O}\left(\tilde
    \lambda^{4/3}\right)$$
  valid when $\delta$ is small enough as compared to $\lambda$.
  $\bar c$ was already computed in the previous section
  and is given by
  $$\bar c = \frac{1-\theta }{\theta } + \frac{\sigma^2\left(1-2
      \theta\right)}{24 \theta^2} \delta + \frac{\left(24 \theta^2-22 \theta
      +5\right) \sigma^4}{2880 \theta^3} \delta^2 +
  \mathcal{O}\left(\delta^{3}\right).$$ Asymptotics for two more coefficients
  are given by
  $$\tilde c_3= \frac{1-\theta}{2 \theta}\left (\frac{6}{\theta (1-\theta) }
  \right )^\frac{1}{3} + \mathcal{O}(\delta^\frac{1}{2}),$$
  $$ \tilde c_4= \frac{(1-\theta)^2}{4 \theta} \left (\frac{6}{\theta
      (1-\theta)} \right)^\frac{2}{3}+ \mathcal{O}(\delta^\frac{1}{2}).$$

  As $\bar s$ is a function of $c$, after plugging, we get
  $$ \bar s=1+ \bar s_1 \lambda ^\frac{1}{3}+\bar s_2 \lambda ^\frac{2}{3}+
  \bar s_3 \lambda+\mathcal{O}\left(\lambda \right), $$ 
  where asymptotics are given by
  $$\bar s_1= \left ( \frac{6}{\theta (1-\theta)} \right)^\frac{1}{3} + 
  \left ( \frac{6(1-\theta)^2}{\theta} \right)^\frac{1}{3} \sigma
  \delta^\frac{1}{2}+\mathcal{O}(\delta).$$
  $$\bar s_2=\frac{1}{2} \left (  \frac{6}{\theta (1-\theta)}
  \right)^\frac{2}{3}+
  \frac{(7-2\theta)(1-\theta)\sigma}{4}\left(\frac{6}{\theta (1-\theta)}
  \right)^\frac{2}{3}\delta^\frac{1}{2}+ \mathcal{O}(\delta). $$

  From here on, we just need to follow the proof of
  Theorem~\ref{thr:no-trade-region} and Theorem~\ref{thr:growth-rate},
  respectively, using the new asymptotics derived above for $c$ and
  $\overline{s}$.
\end{proof}

\bibliographystyle{alpha}
\bibliography{Bezirgen}

\appendix

\section{ Proofs of Some Theorems} \label{ProofsAppendix}
\begin{proof}[Proof of Proposition 4.1]
 For $p=\frac{1}{2},$ we have $\bar c=1$ and $F(\bar c)=F(1)=0.$ Moreover, $$F'(c)=-\frac{\log(d)}{1-d}d^{\frac{(c+d)(c-1)}{c(1-d)}} \left
[c^2+d +2 \frac{1-d}{\log(d)} c \right].$$
We see that $F$ is increasing on $[x_1, \infty)$ where $x_{1}= \frac{1-d}{-\log(d)}+\sqrt{(\frac{1-d}{-\log(d)})^2-d}$ 
is the larger root of the parabola $c^2+d +2 \frac{1-d}{\log(d)} c.$
Elementary calculus shows that $1> x_2.$ Hence, we conclude that there exists a unique $c > 1$ s.t. $F(c)=\lambda.$

Now, let $p \neq \frac{1}{2}$. Denote $c_1=\frac{d(1-2p)}{p+pd-d}$ and $c_2=\frac{1-p-p d}{2p-1}$ which are 
the roots of $c(p+pd-d)+d(2p-1)$ and $1-p-d p-c (2 p-1),$ respectively. Moreover, denote 
\begin{equation} \label{rfunc}
r(c)=\frac{[c(p-d+pd)+d(2p-1)][1-p-d p-c (2 p-1)]}{c (1-d)^2(1-p)^2}.
\end{equation}

 We see that $F(\bar c)=0$ and
\begin{align*}
F'(c)&= (2p-1)(p+pd-d) \Bigg( \frac{c^2(1+b) - 
2 c_2 c+(b-1)d \bar c }{\big[1-p-dp-c(2p-1)\big]^{2}}
\Bigg) r(c)^{b}.
\end{align*}

If $\frac{1}{2}<p<\frac{1}{1+d},$ then $c_1<0<\bar c<c_2$ and $b>1.$ Note that $r(c)>0$ for 
$\bar c < c< c_2$ and $r(c) \to 0$ for $c \uparrow c_2.$ 
Hence, we obtain $F(c) \to 1 \mbox{ for } c \uparrow c_2.$
Intermediate value theorem implies that there is a $c$ on $(\bar c, c_2).$ s.t. $F(c)=\lambda.$ 

We see that if $\bar c <c<c_2,$ then the sign of the parabola $c^2(1+b)-2 c_2 c+(b-1)d \bar c$ determines the sign of 
$F'.$ If the parabola has no root, then $F'(c)>0$ for $\bar c <c<c_2.$
Recalling $F(\bar c)=0,$ we conclude that there exists a unique $c$ on $(\bar c, c_2)$ s.t. $F(c)=\lambda.$ 
If the parabola has a root, then  the smaller root $x_1$ satisfies 
$$ x_1 \leq \frac{c_2}{1+b} <\frac{c_2}{c_2+2} \leq \bar c.$$
Hence, depending on whether $\bar c <x_2$ or not, $F$ decreases on 
$(\bar c, x_2)$ and increases on $(x_2, c_2)$ or only increases on $(\bar c, c_2).$ Due to $F(\bar c)=0,$ in both cases, we get that
there exists a unique $c$ on $(\bar c, c_2)$ s.t. $F(c)=\lambda.$

If $\frac{d}{1+d}<p<\frac{1}{2},$ then $c_2<0<c_1<\bar c$ and $b<-1.$
Note that $r(c)>0$ for $c>\bar c$ and 
$$\frac{r(c)}{\frac{c(p+pd-d)+d(2p-1)}{(1-d)(1-p)}} \to \frac{1-2p}{(1-d)(1-p)}>0 \mbox{ for } c \uparrow \infty.$$
Since $b-1<-2,$ we get $F(c) \to 1 \mbox{ for } c \uparrow \infty.$
Now, intermediate value theorem implies that there is a $c \in (\bar c, \infty)$ s.t. $F(c)=\lambda.$ 

If $c>\bar c,$ then the sign of $F'$ is the opposite of the sign of the parabola $c^2(1+b)-2 c_2 c+(b-1)d \bar c$ due to $2p-1<0.$
The leading coefficient of the parabola, $1+b,$ is negative. 
Hence, if the parabola has no root, then $F'(c)>0$ for $c>\bar c.$ Hence, there exists a unique $c$
on $(\bar c, \infty)$ s.t. $F(c)=\lambda.$ If the parabola has a root, then the smaller root $x_1$ satisfies
$$ x_1 \leq \frac{c_2}{1+b} \leq \bar c, $$
where the last inequality follows due to the fact that the function $w(z)=\frac{z \log(z) }{z-1}$ is increasing on $(1, \infty)$
and hence $w(\frac{1-p}{p}) \leq w(u).$ Hence, by the same argument as in the previous case, 
we obtain that there exist a unique root on $(\bar c, \infty).$
\end{proof}


\begin{proof}[Proof of Proposition 4.2]
 For $p=\frac{1}{2},$ we know from Proposition \ref{resultsc} that $c > 1.$ As $\frac{(c+d)(c-1)}{c(1-d)}$ is strictly increasing
for $c>0,$ we get $k >0.$

Now let $p \neq \frac{1}{2}.$ Recall, how we defined $r(c).$ Differentiation yields
$$ r'(c)=\frac{(1-2p)}{(1-d)^2(1-p)^2} \frac{(p+pd-d)c^2+d(1-p-pd)}{c^2}.$$

If $\frac{1}{2}<p<\frac{1}{1+d},$ then $c_1<0<\bar c<c<c_2.$ We see that $r$ is strictly decreasing 
and positive function on $(0, c_2).$ This implies $0<r(c)<r(\bar c)=1.$ As $p>\frac{1}{2}, $ we get $k > 0.$

If $\frac{d}{1+d}<p<\frac{1}{2},$ then we note that $c_2<0<c_1<\bar c<c.$ Since $r$ is strictly increasing and positive 
on $[\bar c, \infty)$ we get $r(c)>r(\barc)=1.$ Due to $p<\frac{1}{2},$ we obtain $k>0.$
\end{proof}

\begin{proof}[Proof of Proposition 4.4]

We assume $p \neq \frac{1}{2},$ the other case being similar. To start with, we shall prove that $g$ is well-defined on $\{ d,1, \ldots, \bar s, u \bar s \}.$
If $\frac{d}{1+d}<p<\frac{1}{2},$ then $\frac{1-p}{p}>1$ and hence the numerator satisfies 
\begin{align*}
-\left(1- (\f{1-p}{p})^{-\f{\log(s)}{\log d}}\right) + \beta_p & \leq 
(\f{1-p}{p})^{k+1}+\beta_p-1=\frac{1-p}{p}r(c)+\frac{c(2p-1)+p+pd-1}{(1-d)(1-p)} \\
& =\frac{1-p-pd-c(2p-1)}{(1-d)(1-p)} \frac{(c+1)d(2p-1)}{c p (1-d)} <0,
\end{align*}
which shows that $g$ is well-defined. If $\frac{1}{2}<p<\frac{1}{1+d},$ then $\frac{1-p}{p}<1$ and so we get
\begin{align*}
-\left(1- (\f{1-p}{p})^{-\f{\log(s)}{\log d}}\right) + \beta_p & \geq 
(\f{1-p}{p})^{k+1}+\beta_p-1=\frac{1-p}{p}r(c)+\frac{c(2p-1)+p+pd-1}{(1-d)(1-p)} \\
& =\frac{1-p-pd-c(2p-1)}{(1-d)(1-p)} \frac{(c+1)d(2p-1)}{c p (1-d)}>0,
\end{align*}
where $1-p-pd-c(2p-1)>0$ is due to $c<\frac{1-p-pd}{2p-1}$ (recall Proposition \ref{resultsc}). As a result, we obtain that
$g$ is well-defined.

To show that $g$ is increasing, we calculate 
$$g'(s)=\frac{(c+1)\beta_p (\frac{1-p}{p})^\frac{-\log(s)}{\log(d)} \log\left(\frac{1-p}{p}\right)}
{s \log(d)\left[-\left(1-(\frac{1-p}{p})^\frac{\log(s)}{\log(u)} \right)+\beta_p\right]^2}. $$
Here, the denominator is negative since $\log(d)<0.$ The sign of the numerator
depends on the signs of $\beta_p$ and $\log\left(\frac{1-p}{p}\right).$ We easily check that 
$\beta_p \log\left(\frac{1-p}{p}\right)$
is negative for both cases $\frac{d}{1+d}<p<\frac{1}{2}$ and 
$\frac{1}{2}<p<\frac{1}{1+d}.$ Therefore, we conclude that $g$ is increasing.

Moreover, we observe that elementary calculation shows that $g$ indeed satisfies the ``smooth pasting'' conditions after plugging
the values for $c$, $\bar s$ and $\beta_p.$

Now, we show that $ (1-\lambda) s \leq g(s) \leq s \mbox{ for } 1 \leq s \leq \bar s.$ Define $ H(s)=\frac{g(s)}{s}.$
Since $H(1)=1$ and $H(\bar s)=g(\bar s)/ \bar s=1-\lambda,$ it is enough to prove that $H$ is decreasing. Calculation yields
$$H'(s)=\frac{c (\frac{1-p}{p})^\frac{-2\log(s)}{\log(d)}+\left[(c+1) \beta_p \left(\frac{\log(\frac{1-p}{p})}{\log(d)}\right )+
\beta_p(c-1)-2c \right](\frac{1-p}{p})^\frac{-\log(s)}{\log(d)}-(\beta+c)(\beta-1)}
{s^2\left[-\left(1-(\frac{1-p}{p})^\frac{\log(s)}{\log(u)} \right)+\beta_p\right]^2 }.$$
The denominator is positive, hence it suffices to show that the numerator is negative. 
We observe that the numerator is a parabola
in $(\frac{1-p}{p})^\frac{-\log(s)}{\log(d)}$ with positive leading coefficient $c$. 
Thus, the numerator attains its maximum value at the boundaries of $[1, \bar s].$ Denoting the numerator by $N(s),$ 
we obtain that $N(\bar s)= r(c) N(1)$ where
$$N(1)= \beta_p  \left ((c+1)\left(\frac{\log(\frac{1-p}{p})}{\log(d)} -\frac{2p-1}{(1-p)(1-d)}\right )+\frac{2p-1}{1-p} \right). $$
Since $r(c)>0,$ we are done if we show that $N(1)<0.$ If $\frac{1}{2}<p<\frac{1}{1+d},$ then we obtain 
$\frac{\log(\frac{1-p}{p})}{\log(d)}<\frac{2p-1}{p(1-d)}$ since the function 
$\frac{\log(z)}{z-1}$ is decreasing on $(0,1)$ and $0<d<\frac{1-p}{p}<1.$ Combining this with 
$c< \frac{1-p-pd}{2p-1}$ (recall Proposition \ref{resultsc}), we obtain $N(1)<0.$ If $\frac{d}{1+d}<p<\frac{1}{2},$ then by 
similar arguments we get $\frac{\log(\frac{1-p}{p})}{\log(d)}>\frac{(2p-1)d}{p(1-d)}.$ 
Recalling from Proposition \ref{resultsc} that $c>\bar c$, we again obtain $N(1)<0.$

Lastly, denoting $n=\frac{\log(s)}{\log(u)}$ and $x=\frac{1-p}{p},$ we obtain
\begin{align*}
\frac{p \frac{g(us)}{g(s)} +(1-p)\frac {g(ds)}{g(s)}-1}{(\frac{g(us)}{g(s)}-1)(1-\frac{g(ds)}{g(s)})}=&
\frac{[p (g(us)-g(s))+(1-p) (g(ds)-g(s))]g(s)}{[g(us)-g(s)][g(s)-g(ds)]} \\ 
=&\frac{[p \frac{(c+1) \beta_p(x^n-x^{n+1})}{(x^{n+1}+\beta_p-1)(x^n+\beta_p-1)}+
(1-p)\frac{(c+1) \beta_p(x^n-x^{n-1})}{(x^{n-1}+\beta_p-1)(x^n+\beta_p-1)} ] g(s) }
{\frac{(c+1) \beta_p(x^n-x^{n+1})}{(x^{n+1}+\beta_p-1)(x^n+\beta_p-1)} \frac{-(c+1) \beta_p(x^n-x^{n-1})}{(x^{n-1}+\beta_p-1)(x^n+\beta_p-1)} } \\
=& \frac{(x^n+\beta_p-1) g(s)}{(c+1) \beta_p}=\frac{g(s)}{c+g(s)}.\qedhere 
\end{align*}
\end{proof}

\begin{proof}[Proof of Lemma~\ref{lem:1}]
  Recall that $c$ was constructed as root of an equation $F(\delta,c) =
  \lambda$, cf.~Proposition~\ref{resultsc}. Note that $F(\delta,c)$ is
  continuous in $\delta$ and its limit at $\delta = 0$, which we shall denote
  by $F(0,c)$ determines the optimal wealth fraction in the Black-Scholes
  model by $F(0,c) = \lambda$. 

  We set $H(\delta,c) \coloneqq (\delta, F(\delta,c))$, which we consider as
  $H:M \to N$, where $M \subset [0,\infty[ \times \R$ and $N \subset
  [0,\infty[ \times [0,\infty[$, as we do not allow for negative transaction
  costs. We fix $(\delta, \lambda) \in N$ and set $c = G(\delta,
  \lambda)$. First we are going to construct a compact set $U \subset M$ such
  that the interior of $U$ is a neighborhood of $(\delta,c)$ and the interior
  of $H(U)$ is a neighborhood of $(\delta,\lambda)$.

  Choose some $0 \le \epsilon_1 \le \delta < \epsilon_2$, with the
  understanding that we require $\epsilon_1 < \delta$ whenever $\delta>0$ and
  define
  \begin{equation*}
    U \coloneqq \bigcup_{\delta \in [\epsilon_1, \epsilon_2]} \set{\delta}
    \times [c-\epsilon,c+\epsilon],
  \end{equation*}
  for some $\epsilon>0$ small enough that $U \subset M$. Since
  $F(\delta^\prime,\cdot)$ is continuous, $F(\delta^\prime,
  [c-\epsilon,c+\epsilon])$ is a compact interval $[l(\delta^\prime),
  u(\delta^\prime)]$, where $l$ and $u$ are continuous functions of $\delta$
  by uniform continuity. Now choose $\eta > 0$ small enough that $l(\delta) <
  \lambda - 2\eta$ and $u(\delta) > \lambda + 2\eta$. By continuity of $l$ and
  $u$, we can find $0\le \epsilon_1 \le \kappa_1 \le \delta < \kappa_2 <
  \epsilon_2$ with the understanding that $\kappa_1 < \delta$ unless $\delta =
  0$, such that
  \begin{equation*}
    \forall \kappa_1 \le \delta^\prime \le \kappa_2: \ \abs{l(\delta^\prime) -
    l(\delta)} < \eta, \ \abs{u(\delta^\prime) - u(\delta)} < \eta.
  \end{equation*}
  This implies that the closed ball $B(\lambda,\eta)$ with radius $\eta$
  around $\lambda$ is contained in every $F(\delta^\prime,
  [c-\epsilon,c+\epsilon])$ with $\kappa_1 \le \delta^\prime \le \kappa_2$,
  showing that the ball $[\kappa_1,\kappa_2] \times B(\lambda,\eta)$ is
  contained in $H(U)$, implying that the interior of $H(U)$ is a neighborhood
  of $(\delta, \lambda)$. (When $\delta = 0$, then the interior is understood
  in the sense of the topology on $[0,\infty[\times [0,\infty[$. We have
  tacitly assumed that either $\lambda > 0$ or we also allow for negative
  transaction costs.)

  Now fix a sequence $(\delta_n, \lambda_n) \xrightarrow{n\to\infty}
  (\delta,\lambda)$. We may assume that $(\delta_n, \lambda_n) \in H(U)$, and
  denote $(\delta_n, c_n) \coloneqq H^{-1}(\delta_n, \lambda_n) \in U$. By
  closedness, every converging subsequence of $(\delta_n, c_n)$ must have a
  limit in $U$. Let $(\delta_{n_k}, c_{n_k})$ be such a converging subsequence
  with limit $(\delta, c^\prime)$. Then, by continuity of $H$, we have
  $H(\delta, c^\prime) = (\delta, \lambda)$, implying that $c^\prime = c
  \coloneqq G(\delta,\lambda)$.
\end{proof}

\end{document}